\documentclass[a4paper,UKenglish]{lipics-v2018}
\usepackage{microtype}
\usepackage{amsmath}
\usepackage{amssymb}
\title{Reconfiguring graph minors}

\author{Benjamin Moore}{Department of Combinatorics and Optimization, University of Waterloo, Canada}{brmoore@uwaterloo.ca}{}{}

\author{Naomi Nishimura}{David R. Cheriton School of Computer Science, University of Waterloo, Canada}{nishi@uwaterloo.ca}{}{}

\author{Vijay Subramanya}{David R. Cheriton School of Computer Science, University of Waterloo, Canada}{v7subram@uwaterloo.ca}{}{}

\authorrunning{B. Moore, N. Nishimura, and V. Subramanya}

\Copyright{Benjamin Moore, Naomi Nishimura, and Vijay Subramanya}

\subjclass{G.2.2 Graph Theory}

\keywords{reconfiguration, graph minors, graph algorithms}

\relatedversion{}

\funding{Supported by the Natural Sciences and Engineering Research Council of Canada}

\acknowledgements{We wish to thank Tesshu Hanaka, Venkatesh Raman, Krishna Vaidyanathan, and Marcin Wrochna for helpful discussions.}

\EventEditors{John Q. Open and Joan R. Acces}
\EventNoEds{2}
\EventLongTitle{42nd Conference on Very Important Topics (CVIT 2016)}
\EventShortTitle{CVIT 2016}
\EventAcronym{CVIT}
\EventYear{2016}
\EventDate{December 24--27, 2016}
\EventLocation{Little Whinging, United Kingdom}
\EventLogo{}
\SeriesVolume{42}
\ArticleNo{23}

\nolinenumbers 
\hideLIPIcs  

\begin{document}

\maketitle

\newtheorem{observation}[theorem]{Observation}

\newcommand{\fillin}[1]{\textcolor{red}{\tt #1}}

\begin{abstract}
Under the reconfiguration framework, we consider the various ways that a target graph $H$ is a {\em minor} of a host graph $G$, where a subgraph of $G$ can be transformed into $H$ by means of {\em edge contraction} (replacement of both endpoints of an edge by a new vertex adjacent to any vertex adjacent to either endpoint). Equivalently, an {\em $H$-model} of $G$ is a labeling of the vertices of $G$ with the vertices of $H$, where the contraction of all edges between identically-labeled vertices results in a graph containing representations of all edges in $H$.

We explore the properties of $G$ and $H$ that result in a connected {\em reconfiguration graph}, in which nodes represent $H$-models and two nodes are adjacent if their corresponding $H$-models differ by the label of a single vertex of $G$.  Various operations on $G$ or $H$ are shown to preserve connectivity.
In addition, we demonstrate properties of graphs $G$ that result in connectivity for the target graphs $K_2$, $K_3$, and $K_4$, including a full characterization of graphs $G$ that result in connectivity for $K_2$-models, as well as the relationship between connectivity of $G$ and other $H$-models.

\end{abstract}

\section{Introduction}\label{sec-intro}

Graph minors have been studied extensively as a means for categorizing graphs and exploiting their properties.  A graph $H$ is a {\em minor} of a graph $G$ if $H$ can be formed from a subgraph of $G$ by a series of edge contractions, where the {\em contraction} of the edge $uv$ results in the replacement of both $u$ and $v$ by a new vertex $w$ that is adjacent to any vertex that was adjacent to $u$ or $v$ (or both).  Much of the research in the area has focused on classes of graphs that are closed under the taking of minors, and on exploiting properties of graphs known not to contain certain graphs as minors. For example, it is known that for every minor closed class, that class is characterized by a finite set of forbidden minors \cite{RobertsonSeymour}. Additionally, it has been shown that for any fixed graph $H$, every $H$-minor-free graph of treewidth $w$ has an $\Omega(w) \times \Omega(w)$ grid as a minor~\cite{DH05}.

In our work, we instead focus on the solution space of $H$-models of a graph $G$ using the reconfiguration framework~\cite{IDHPSUU11,N17,H13}, where an $H$-model is a mapping that labels the vertices of $G$ with the vertices of $H$.  A {\em reconfiguration graph} for an instance of a problem consists of a node for each possible feasible solution and an edge between any two nodes representing solutions that are adjacent.  The definition of adjacency may be presented as a {\em reconfiguration step} used to transform a solution into a neighbouring solution. Structural properties of the reconfiguration graph, including its diameter and whether it is connected, are of interest both in their own right and as keys to solving algorithmic problems, such as determining whether there is a path (or {\em reconfiguration sequence}) in the graph between two given vertices  and, if so, finding the shortest such path.  

In this paper, we consider how the connectivity of the reconfiguration graph depends on the choices of the host $G$ and target $H$.  We consider an instance of \textsc{Minor Reconfiguration} to consist of a {\em host graph} $G$ and {\em target graph} $H$ such that $H$ is a minor of $G$.  Each node in the reconfiguration graph consists of a labeling of the vertices of $G$ with the vertices of $H$ (or, more simply, integers in $\{1, \ldots, |V(H)|\}$) such that the contraction of each edge with identically-labeled endpoints results in a graph that, upon deletion of zero or more edges, yields $H$.  We consider two $H$-models to be adjacent if they differ by a single label.

Although we are the first to consider the reconfiguration of minors, several papers have considered the reconfiguration of  subgraphs~\cite{subgraph,subgraphmoritz}.  The representation of a configuration as a labeling of the vertices has been used for problems entailing moving labels from a source to a target configuration using the minimum number of swaps, where labels (or tokens) on adjacent vertices can be exchanged (detailed in a survey of reconfiguration~\cite{N17}), and labeled edges have been considered in the reconfiguring of triangulations~\cite{annafliplabel}.

We begin by establishing properties of $k$-connected graphs and minors
in Section~\ref{sec-prelims}, based on which we form a toolkit of
techniques used in reconfiguration (Section~\ref{sec-tools}).
We consider various properties of $G$ and $H$ that determine whether or not the reconfiguration graph is connected. For a target graph $H$, we define $\mbox{host}(H)$ to be the set of host graphs $G$ such that the reconfiguration graph for $G$ and $H$ is connected.
We then
focus on characterizing $\mbox{host}(K_2)$
(Section~\ref{sec-k2}),
$\mbox{host}(K_3)$ (Section~\ref{sec-k3}), and
$\mbox{host}(K_4)$ (Section~\ref{sec-k4}).  Finally, in
Section~\ref{sec-conclusions}, we summarize the results and present
directions for future work.

\section{Preliminaries}\label{sec-prelims}

We define key terms used in the description of graphs; for common terms not defined in this paper, the reader is referred to a resource on graph theory~\cite{D05}.  We will frequently focus on subsets of the vertices; 
for a subgraph $V \subseteq V(G)$, the {\em induced subgraph} $G[V]$ is the subgraph with vertex set $V$ and edge set $\{uv \in V(G) \mid u, v \in V\}$.  As shorthand, for $G$ a graph and $S$ a set of vertices, we use $G \setminus S$ to denote $G[V(G) \setminus S]$.  In order to avoid confusion with the vertices in graphs $G$ and $H$, we refer to the {\em nodes} of a reconfiguration graph.

\subsection{Properties of $k$-connected graphs}

We focus on various ways of connecting vertices in the graph. 
A {\em cut set} $S$ of $G$ is a set of vertices such that $G \setminus S$ consists of at least two components; the member of a cut set of size one is also called a {\em cut vertex}.  A {\em bridge} is an edge whose deletion disconnects the graph. A graph is $k$-connected if there is no cut set of size $k$.  Equivalently, in a $k$-connected graph there exist $k$ vertex-disjoint paths between any pair of vertices in the graph.  At times we will focus on how highly connected a specific vertex might be.  A {\em universal vertex} is adjacent to all other vertices in the graph.  In a {\em complete graph} on $j$ vertices, denoted $K_j$, all vertices are universal vertices.

To characterize the behaviour of various host and target graphs, we make use of characterizations of graphs in terms of a base graph class and a series of operations. {\em Adding an edge} consists of adding an edge between two vertices in $V(G)$.  To {\em split} a vertex $v$ is to first delete $v$ from $G$, and then add two vertices $v_{1}$ and $v_{2}$ to $G$ such that $v_{1}v_{2} \in E(G)$, each neighbour of $v$ in $G$ is a neighbour of exactly one of $v_{1}$ or $v_{2}$, and $\deg(v_{i}) \geq 3$ for $i \in \{1,2\}$.

We make use of Tutte's characterization of $3$-connected graphs and Ding and Qin's characterization of a subset of the $4$-connected graphs, given below.  The base case for Tutte's characterization is a wheel, defined as follows.

\begin{definition}\label{def-wheel}
A {\em $k$-wheel} $W_k$ is a graph on $k+1$ vertices, the {\em rim vertices} $r_1, \ldots, r_k$ and the {\em hub vertex} $h$, where there is a cycle induced on the rim vertices and an edge between $h$ and each rim vertex.
\end{definition}

\begin{theorem}~\cite{wheelsandwhirls}\label{thm-tutte-3conn}
A graph is $3$-connected if and only if it is obtained from a wheel by repeatedly adding edges and splitting vertices.
\end{theorem}

To state Ding and Qin's result, we need a few additional definitions. The \emph{line graph} $L(G)$ of a graph $G$ has a vertex corresponding to each edge of $G$ and two vertices are adjacent if their corresponding edges share an endpoint in $G$. A graph is \emph{cubic} if each vertex has degree three. Furthermore, a cubic graph with at least six vertices is \emph{internally $4$-connected} if its line graph is $4$-connected. One of the base classes for their characterization is a square of a cycle, as defined below.

\begin{definition}
The {\em square of a cycle} $C^2_k$ is formed from the cycle $C_k$ by adding an edge between any pair of vertices joined by a path of length two.
\end{definition}

Finally, we say a sequence of $4$-connected graphs $G_{1},\ldots,G_{n}$ form a \emph{$(G_{1},G_{n})$-chain} if for all $i \in \{1,\ldots,n-1\}$, there exists an edge $e$ such that $G_{i+1}$ is formed from $G_{i}$ by removal of the edge $e$.
Theorem \ref{thm-dingqin} is a generalization of a well-known theorem of Martinov \cite{Martinov}.

\begin{theorem}~\cite{DQ15}\label{thm-dingqin}
  Let $\mathcal{C} = \{C^2_k : k \geq 5\}$ and $\mathcal{L} = \{H : H$ is the line graph of an internally $4$-connected cubic graph$\}$. Let $G$ be a $4$-connected graph not in $\mathcal{C} \cup \mathcal{L}$. Then if $G$ is planar, there is a $(G,C_{6}^{2})$-chain. Otherwise, there is a $(G,K_{5})$-chain. 
\end{theorem}

\subsection{Branch sets, $H$-models, and block trees}

For the purposes of reconfiguration, we make use of an equivalent definition of a minor as a mapping of each vertex of host graph $G$ to a vertex of target graph $H$. For convenience, we sometimes represent the vertices of $H$ as integer labels.
\begin{definition}\label{def-branchsets}
  For graphs $G$ and $H$ and mapping $f:V(G) \rightarrow V(H)$, we refer to $f(v)$ as the {\em label} of $v$ and define the {\em branch set} $G(f,i)$ to be the subgraph of $G$ induced on the set of vertices with label $i$.
\end{definition}

For ease, we will make use of $|G(f,i)|$ to denote $|V(G(f,i))|$. Given a mapping between $V(G)$ and $V(H)$, an edge of $G$ is a {\em connecting edge} if its endpoints are members of two different branch sets, and we say that it {\em connects} those two branch sets.  A mapping is equivalent to a minor when two additional properties hold, as indicated in Definition~\ref{def-model}.

\begin{definition}\label{def-model}
  For graphs $G$ and $H$ and mapping $f:V(G) \rightarrow V(H)$, we say that $H$ is a {\em minor of $G$} and that $f$ is an {\em $H$-model of $G$} if the following conditions hold:  
\begin{enumerate}
\item{}\label{model-vertices} for each $i \in V(H)$, each branch set $G(f,i)$ is nonempty and connected, and
\item{}\label{model-edges} for each edge $ij \in E(H)$, there exists an edge in $E(G)$ connecting $G(f,i)$ and $G(f,j)$.
\end{enumerate}
\end{definition}

 We will often find it convenient to view each branch set in terms of the tree structure of its $2$-connected components.  A {\em block} of a connected graph is either a maximal $2$-connected subgraph or one of the endpoints of a bridge. The {\em block tree} of a connected graph consists of a node for each block $B$; there is an edge between the nodes corresponding to blocks $B$ and $B'$ if there exists a cut vertex $v$ of $G$ such that $V(B) \cap V(B') = \{v\}$.  Given a graph $G$, an $H$-model $f$, and a label $a$, we use $T(G, f, a)$ to denote the block tree for $G(f,a)$.  In addition, for a subgraph $A$ of $G$, we use $T(G,f,a,A)$ to denote the subtree of $T(G, f, a)$ induced by the blocks containing vertices in $V(G(f,a)) \cap V(A)$. For convenience, we sometimes use "block of $G(f,a)$" to refer to a block of $T(G, f, a)$.

To make use of the tree structure, our algorithms typically process a block tree starting with blocks that are leaves of their block trees, or {\em leaf blocks}; a branch set that is $2$-connected can be viewed as having a block tree consisting of a single leaf block.  For ease of description, we will refer to the cut vertices of $G$ that appear in multiple blocks as {\em joining vertices} and all other vertices as {\em interior vertices}.

\subsection{Essential edges, crucial vertices, weak connections, and lynchpins}

When considering how labelings can be reconfigured, we need to ensure that we retain the connecting edges as required in Definition~\ref{def-model}.  In doing so, we need to pay particular attention to vertices and edges whose relabeling will cause problems. 

When there exists only a single edge that connects a pair of branch sets with labels $a$ and $b$, $ab \in  E(H)$, we call such an edge an {\em essential edge}, and denote it as $\mbox{ess}(a,b)$.  If all the edges between branch sets with labels $a$ and $b$ have the same endpoint in $a$, we call that vertex an {\em essential vertex for $b$}; clearly every endpoint of an essential edge is an essential vertex, but not every essential vertex is the endpoint of an essential edge.

The presence of essential vertices will be important in determining when it is easy to relabel vertices.  For any two labels, if the branch set with label $a$ contains an essential vertex for $b$ or if the branch set with label $b$ contains an essential vertex for $a$, we will say that the branch sets with labels $a$ and $b$ are {\em weakly connected}, or form a {\em weak connection}.  

Our results rely on the interplay between the presence of weak connections and the connectivity of a graph. For each weak connection, we identify a vertex as the {\em lynchpin} for the connection.  When the branch sets with labels $a$ and $b$ are weakly connected by an essential edge, then either of the endpoints of the essential edge can be designated as the lynchpin.  Otherwise, the (single) essential vertex giving rise to the weak connection is the lynchpin for that connection.  We will use lynchpins to form cut sets between non-lynchpins and other branch sets.

A vertex $v$ with label $a$ is a {\em crucial vertex} if it is an essential vertex for $b$ and an essential vertex for $c$, for $b \ne c$, and a {\em non-crucial vertex} otherwise.  If for some distinct labels $a$, $b$, and $c$, a vertex $v \in G(f,a)$ is essential for $c$ and also has at least one neighbour in $G(f,b)$, then $v$ is a {\em $b$-crucial vertex}.  Clearly, a vertex in $G(f,a)$ that is essential for $b$ and $c$ is crucial, $b$-crucial, and $c$-crucial.

The following observation characterized non-crucial vertices.

      \begin{observation}\label{obs-non-crucial}
For any non-crucial vertex $v \in G(f,a)$ that has at least one neighbour in a different branch set, there exists at least one label $b \ne a$ such that $v$ has a neighbour in $G(f,b)$ and $v$ is not $b$-crucial. 
  \end{observation}

  \begin{proof}
Suppose to the contrary that $v$ is a non-crucial vertex but for every label $b$ such that $v$ has a neighbour in $G(f,b)$, $v$ is $b$-crucial.  By the definition of $b$-crucial, $v$ is essential for some $c \ne b$.  By our assumption, since $v$ has a neighbour in $G(f,c)$, $v$ is $c$-crucial.  By the definition of $c$-crucial, $v$ is essential for some $d \ne c$.  We can then conclude that $v$ is  essential for two different labels, and hence is crucial, forming a contradiction. 
\end{proof}

\subsection{Properties of $H$-models of $k$-connected graphs}\label{sec-model-k-conn}

When $G$ is $k$-connected, we are able to establish properties of connecting edges of branch sets, as shown in Lemma~\ref{lemma-leafblock} and Lemma~\ref{lemma-two-weak}, as well as the structure of weak edges (Lemma~\ref{lemma-all-weak}). The results make use of the fact that in a $k$-connected graph there cannot be a cut set of size less than $k$ separating any two vertices; cut sets are typically formed from the joining vertices of leaf blocks and lynchpins, and vertices separated by cut sets are typically non-lynchpins.  

\begin{lemma}\label{lemma-leafblock}
Given a $k$-connected graph $G$ and an $H$-model $f$ of $G$ for some graph $H$, for any branch set $G(f,a)$, each leaf block has $k-1$ interior vertices that are endpoints of connecting edges.
\end{lemma}

\begin{proof}
  Due to $k$-connectivity, there must be $k$ internally-vertex-disjoint paths between any pair of vertices. In particular, there must be $k$ internally-vertex-disjoint paths between an interior vertex in a leaf block of $G(f,a)$ and a vertex in another branch set.  At most one path can use the joining vertex of the leaf block.  The remaining $k-1$ paths must then make use of distinct interior vertices of the leaf block, each of which must have an edge to a vertex in another branch set.
  
\end{proof}

\begin{lemma}\label{lemma-two-weak}
  Given a $k$-connected graph $G$ and an $H$-model $f$ of $G$, where $|V(H)|=k$, suppose there exist branch sets $G(f,\ell)$ and $G(f,m)$ such that ${\ell}m \in E(H)$ and there are weak connections between $G(f,\ell)$ and each branch set other than $G(f,m)$ (where a weak connection between $G(f,\ell)$ and $G(f,m)$ is possible but not required). Then, the following hold:
  \begin{enumerate}
  \item{} Each leaf block in $G(f,\ell)$ must contain an interior vertex that is the endpoint of a connecting edge to $G(f,m)$.
    \item{} If it is possible to designate lynchpins of the weak connections such that $G(f,\ell)$ contains a non-lynchpin, then each leaf block in $G(f,m)$ must contain an interior vertex that is the endpoint of a connecting edge to $G(f,\ell)$.
  \end{enumerate}
\end{lemma}

  \begin{proof}
    To see why the first point holds, suppose instead that no such interior vertex existed in a leaf block of $G(f,\ell$).  Then, each path between an interior vertex $u$ in the leaf block in $G(f,\ell)$ and any vertex $v$ in $G(f,m)$ must pass through either the joining vertex of the leaf block or one of the lynchpins for the weak connections.  However, $u$ and $v$ are thus separated by a cut set of size at most $k-1$, contradicting the $k$-connectivity of $G$.

    The argument for the second point is similar; we can show that the joining vertex of the leaf block in $G(f, m)$ and the lynchpins of the weak connections form a cut set of size at most $k-1$ separating any interior vertex in the leaf block and the non-lynchpin in $G(f,\ell)$. 
\end{proof}

\begin{lemma}\label{lemma-all-weak}
Given a $k$-connected graph $G$ and a $K_k$-model $f$ of $G$ such that there is a branch set $B$ with weak connections to all other branch sets, it is not possible to designate lynchpins such that $B$ contains at least one vertex $x$ that is not a lynchpin for any of the weak connections, and at least one other branch set contains a vertex $y$ that is not a lynchpin for any of the weak connections.
\end{lemma}

\begin{proof}
It will suffice to show that if it is possible to designate lynchpins in the way described in the statement of the lemma, then we obtain a contradiction.  To do so, we will show that $x$ and $y$ can be separated by a cut set of size less than $k$, violating the $k$-connectivity of $G$.

Because $x$ is not a designated lynchpin, any path from $x$ to another branch set must pass through one of the designated lynchpins to reach $y$.
Thus, the designated lynchpins associated with the $k-1$ weak connections between $B$ and the remaining branch sets forms a cut set of size at most $k-1$, completing the proof of the lemma. 
\end{proof}

\section{Toolkit for reconfiguration of minors}\label{sec-tools}

In this section, we introduce techniques and properties that are exploited in the results found in the rest of the paper.  In particular, we focus on the types of steps used in reconfiguration and the properties that need to be satisfied for each type of transformation.  In Lemmas~\ref{lemma-conds} and \ref{lemma-universal} we determine conditions under which a vertex can be relabeled in a single step.  In the remainder of the section, we present results that can be used to handle more complex situations in which one or more of the conditions do not hold. 

Lemma~\ref{lemma-conds} delineates the conditions necessary for a vertex to be able to be relabeled from $a$ to $b$ in a single reconfiguration step: it cannot be the only vertex with label $a$, it cannot be a cut vertex in its branch set, it must be connected to a vertex with label $b$, and it is not incident with every edge between the branch sets for labels $a$ and $c$, where $c \ne b$.

\begin{lemma}\label{lemma-conds}
  Given a graph $G$ and an $H$-model $f$ of $G$, a vertex $v$ can be relabeled from $a$ to $b$ in a single reconfiguration step if and only if the following conditions hold:
\begin{enumerate}
  \item{}\label{cond-nonempty} $|G(f,a)| > 1$;
  \item{}\label{cond-notcut} $v$ is not a cut vertex in $G(f,a)$;
  \item{}\label{cond-nbr} $v$ has at least one neighbour in $G(f,b)$; and    
  \item{}\label{cond-edges} $v$ is not a $b$-crucial vertex.
\end{enumerate}    
\end{lemma}

  \begin{proof}
We define a mapping $g$ such that $g(v) = b$ and for every $u \ne v$, $g(u) = f(u)$.  We first show that if any of the conditions do not hold, then $g$ is not an $H$-model.  If condition~\ref{cond-nonempty} is violated, then in $g$ there is no vertex with label $a$, and if condition~\ref{cond-notcut} is violated, $G(g,a)$ is not connected. If condition~\ref{cond-nbr} is violated, then $G(g,b)$ is not connected.  Finally, if condition~\ref{cond-edges} does not hold, there exists a label $c$ such that $ac \in E(H)$ but there is no edge between $G(g,a)$ and $G(g,c)$.

To complete the proof, we now observe that if all the conditions hold, then $g$ is an $H$-model.  In particular, each branch set is nonempty and connected and for each edge $cd$ in $E(H)$, there exists at least one edge between $G(g,c)$ and $G(g,d)$. 
\end{proof}

Because several of the conditions hold automatically for a universal vertex, the following lemma lists a smaller number of conditions:

\begin{lemma}\label{lemma-universal}
  Given a graph $G$ and an $H$-model $f$ of $G$, a vertex $v$ can be relabeled from $a$ to $b$ in a single reconfiguration step if the following conditions hold:
  \begin{enumerate}
  \item{}\label{universal-nonempty} $G(f,a)$ contains a universal vertex $u$ such that $u \ne v$; and
  \item{}\label{universal-nbr} $v$ has at least one neighbour in $G(f,b)$.
  \end{enumerate}
\end{lemma}

  \begin{proof}
It suffices to show that all conditions in Lemma~\ref{lemma-conds} hold.  Condition~\ref{cond-nonempty} follows from the existence of a vertex $u \ne v$ in $G(f,a)$ and condition~\ref{cond-nbr} follows from the existence of a neighbour of $v$ in $G(f,b)$.  The remaining two conditions follow from $u$ being universal: since there are edges from $u$ to each vertex in the branch set, $v$ is not a cut vertex (satisfying condition~\ref{cond-notcut}), and since there are edges from $u$ to each vertex in the graph, condition~\ref{cond-edges} holds. 
\end{proof}

When neither Lemma~\ref{lemma-conds} nor Lemma~\ref{lemma-universal} applies, the relabeling of a vertex requires a series of reconfiguration steps.  When the vertex to be relabeled is the only member of its branch set or a crucial vertex, we first need to fill its branch set with new vertices that can provide the necessary connecting edges to other branch sets.  When the vertex to be relabeled is a cut vertex, we will need to siphon away vertices from its branch set so that it is no longer a cut vertex among the remaining vertices with its label.

When a branch set is a block tree, both filling and siphoning entail the relabeling of vertices in a branch set block by block, starting at the leaf blocks.
If we are able to relabel all the interior vertices of a leaf block, we can simplify the block tree by removing the leaf block.
We will show in Lemma~\ref{lemma-slurp-component} that such relabeling is possible as long as we can avoid certain bad situations involving {\em leaf-crucial models} and {\em leaf-$\ell$-crucial models}, as outlined in Definitions~\ref{def-leaf-crucial-vertex}, \ref{def-limited-vertex}, and \ref{def-hit-bad-vertex}.

\begin{definition}\label{def-leaf-crucial-vertex}
Given a graph $G$ and an $H$-model $f$ of $G$, we say that a vertex $v$ is a {\em leaf-crucial vertex} if $v$ is a crucial vertex that is an interior vertex in a leaf block in its branch set.  An $H$-model that contains a leaf-crucial vertex is a {\em leaf-crucial model}.
\end{definition}

\begin{definition}\label{def-limited-vertex}
  Given a graph $G$ and an $H$-model $f$ of $G$, we say that a vertex $v$ is a {\em leaf-$\ell$-crucial vertex} if $v$ is an $\ell$-crucial vertex that is an interior vertex in a leaf block in its branch set. 
  An $H$-model that contains a leaf-$\ell$-crucial vertex is a {\em leaf-$\ell$-crucial model}.
\end{definition}

\begin{definition}\label{def-hit-bad-vertex}
  Given a graph $G$, an $H$-model $f$ of $G$, a label $a$, and a subgraph $A$ of $G(f,a)$,  we say that {\em $f$ hits a leaf-crucial model on relabeling $A$} if any relabeling of $f$ that changes only the vertices of $A$ can be extended by relabeling only the vertices of $A$ to reach a leaf-crucial model, and that {\em $f$ hits a leaf-$\ell$-crucial model on relabeling $A$} if any relabeling of $f$ that changes only the vertices in $A$ can be extended by relabeling only the vertices of $A$ to reach a leaf-$\ell$-crucial model.
\end{definition}

Observations~\ref{obs-no-crucial} and~\ref{obs-no-limited-vertex} follow from the definitions above.

\begin{observation}\label{obs-no-crucial}
  Given a graph $G$, an $H$-model $f$ of $G$, a label $a$, and a subgraph $A$ of $G(f,a)$, if $f$ does not hit a leaf-crucial model on relabeling $A$, then for each model $g$ reachable from $f$ by relabeling only the vertices of $A$, no interior vertex in a leaf block of $T(G,g,a)$ is crucial.
\end{observation}

  \begin{observation}\label{obs-no-limited-vertex}
  Given a graph $G$, an $H$-model $f$ of $G$, a label $a$, and a subgraph $A$ of $G(f,a)$, if $f$ does not hit a leaf-$\ell$-crucial model on relabeling $A$, then for each model $g$ reachable from $f$ by relabeling only the vertices of $A$, no interior vertex in a leaf block of $T(G,g,a)$ is $\ell$-crucial.
  \end{observation}

In Lemma~\ref{lemma-invariant}, we show that each time we relabel an interior vertex in a leaf block, if the leaf block still has interior vertices, one will be a neighbour of the relabeled vertex.  We use the result in Lemma~\ref{lemma-relabel-leafblock} to show that if we can avoid leaf-crucial models, then it is possible to relabel all the interior vertices in a leaf block (and hence remove it from the branch set).  By repeatedly relabeling leaf blocks, an entire connected component can be siphoned away, as shown in Lemma~\ref{lemma-slurp-component}.

\begin{lemma}\label{lemma-invariant}
  Given a graph $G$ and an $H$-model $f$ of $G$, suppose there exist labels $a$ and $b$ and a vertex $v$ such that $|G(f,a)| \ge 2$, $v$ is in a leaf block $L$ of $T(G,f,a)$, and relabeling $v$ to $b$ (and no other vertices) results in another $H$-model $g$.  Then $v$ has a neighbour in $G(g,a)$, and if $|V(L) \setminus \{v\}|\ge 2$, then $v$ has a neighbour $u$ in $G(g,a)$ such that $u \in V(L)$ and $u$ is an interior vertex in a leaf block of $T(G,g,a)$.  
\end{lemma}

  \begin{proof} We observe that since Lemma~\ref{lemma-conds} holds for the relabeling of $v$ from $a$ to $b$, $|G(f,a)| \ge 2$, and since each branch set is connected, $v$ must then have a neighbour in $G(g,a)$.

  We now suppose that $|V(L) \setminus \{v\}| \ge 2$.  At least one neighbour $u \in V(L)$ of $v$ is in a leaf block $L_g$ of $T(G,g,a)$, as $v$ is not a cut vertex of $G(f,a)$ (by Lemma~\ref{lemma-conds}, condition~\ref{cond-notcut}) and $L$ is a $2$-connected leaf block of $T(G,f,a)$.  If $u$ is an interior vertex in $G(g,a)$, we are done. If instead $u$ is a joining vertex and there exists no interior vertex in $L_g$ that is a neighbour of $v$, then for each interior vertex $w \in V(L_g)$, $u$ lies on every path from $v$ to $w$ in $L$, which contradicts the fact that $L$ is $2$-connected. Hence, $v$ is adjacent to an interior vertex in $G(g,a)$.  
  \end{proof}

\begin{lemma}\label{lemma-relabel-leafblock}
  Given a graph $G$, an $H$-model $f$ of $G$, a label $a$, and a leaf block $L$ of $T(G,f,a)$, suppose that $|V(G(f,a))\setminus V(L)| \geq 1$, $L$ contains at least one interior vertex that is the endpoint of a connecting edge, and $f$ does not hit a leaf-crucial model on relabeling $L$.
  Then we can reconfigure $f$ to a model $g$ such that $g(v)\neq f(v)$ for each $v\in V(L)$ that is an interior vertex of $G(f,a)$ and $g(u) = f(u)$ for all other vertices $u$.
\end{lemma}

  \begin{proof}
  We show that all interior vertices of $L$ can be relabeled.  By assumption, there exists an interior vertex, say $v$, that is the endpoint of a connecting edge, say to $G(f,b)$.  Since $|V(G(f,a))\setminus V(L)| \geq 1$, relabeling $v$ to $b$ satisfies conditions \ref{cond-nonempty}, \ref{cond-notcut}, and \ref{cond-nbr} in Lemma \ref{lemma-conds}.  If condition~\ref{cond-edges} also holds, then we can relabel $v$ to $b$.

If instead condition~\ref{cond-edges} does not hold for the relabeling of $v$ to $b$, then $v$ is a $b$-crucial vertex.  Because $f$ does not hit a leaf-crucial model on relabeling $L$, $v$ is not a leaf-crucial vertex and hence not a crucial vertex. Thus, by Observation~\ref{obs-non-crucial} there exists a label $c\notin \{a,b\}$ such that $v$ has a neighbour labeled $c$ and $v$ is not a $c$-crucial vertex. Hence all conditions of Lemma~\ref{lemma-conds} are satisfied for the relabeling of $v$ to $c$.

As $v$ was an interior vertex in $L$ that could be relabeled, then by Lemma~\ref{lemma-invariant} either all interior vertices have been relabeled, or $v$ has a neighbour that is an interior vertex in $L$.  In the latter case, we can repeatedly apply the same argument until all interior vertices in $L$ have been relabeled.

\end{proof}

\begin{lemma}\label{lemma-slurp-component}
Given a $2$-connected graph $G$ and an $H$-model $f$ of $G$, suppose there exist  $ab\in E(H)$, a cut vertex $x$ of $G(f,a)$, and a connected component $C$ of $G(f,a)\setminus\{x\}$ that contains at least one vertex with a neighbour in $G(f,b)$ such that $f$ does not hit a leaf-crucial model or a leaf-$b$-crucial model on relabeling $C$.  Then we can reconfigure $f$ to a model $g$ such that $g(v)\neq f(v)$ for each $v\in V(C)$, $g(u) = f(u)$ for all $u\notin V(C)$, and $x$ has a neighbour in $G(g,b)$.
\end{lemma}

\begin{proof}
  We use $B$ to denote the block of $T(G,f,a,C)$ containing $x$, and view $T(G,f,a,C)$ as rooted at $B$.  We observe that any leaf block of $T(G,f,a,C)$ is also a leaf block of $T(G,f,a)$.

  To reconfigure $f$ to $g$, we work up the tree $T(G,f,a,C)$ from leaf blocks up to $B$, at each step relabeling all the vertices in the current block with labels different from $a$.  Specifically, if a leaf block does not have an interior vertex with a neighbour labeled $b$, then in Case 1 below, we can relabel the vertices in the block; such a relabeling removes the block from the branch set for label $a$.  If instead a leaf block does have an interior vertex with a neighbour labeled $b$, then in Case 2 below, we can relabel the block with $b$.  Such a relabeling not only removes the block from the branch set for label $a$, but also ensures that the joining vertex of the block has a neighbour with label $b$.  Repeated applications of the two cases suffice to ensure that we eventually reach a point in the process at which $B$ is a leaf block and contains an interior vertex with a neighbour labeled $b$; using Case 2, we can then satisfy the statement of the lemma by ensuring that every vertex in $B$ except $x$ receives label $b$.

\smallskip
\noindent{\bf Case 1}: A leaf block $L$ of $T(G, f, a, C)$ does not contain an interior vertex with a neighbour in $G(f,b)$.
\smallskip

Since $G$ is 2-connected, by Lemma~\ref{lemma-leafblock}, $L$ has at least one interior vertex that is an endpoint of a connecting edge. Because $f$ does not hit a leaf-crucial model and $|V(G(f,a))\setminus V(L)| \geq 1$, by Lemma \ref{lemma-relabel-leafblock}, we can relabel the interior vertices of $L$.

\smallskip
\noindent{\bf Case 2}: A leaf block $L$ of $T(G, f, a, C)$ contains an interior vertex $v$ with a neighbour in $G(f,b)$.
\smallskip

We first observe that we can relabel $v$ to $b$: since $f$ does not hit a leaf-$b$-crucial model on relabeling $C$, $v$ is not $b$-crucial (Observation~\ref{obs-no-limited-vertex}), and hence all the conditions of Lemma~\ref{lemma-conds} hold.  We can then repeat the same argument on the resulting model $h$, as follows. 
For $L_h$ the leaf block of $T(G,h,a,C)$ such that $V(L_h) = V(L)\setminus \{v\}$, if one exists, by Lemma~\ref{lemma-invariant}, $L_h$ contains an interior vertex $u$ that is a neighbour of $v$.  We can then use the fact that $f$ does not hit a leaf-$b$-crucial model on relabeling $C$ to again apply Observation~\ref{obs-no-limited-vertex} to conclude that $u$ is not $b$-crucial and that Lemma~\ref{lemma-conds} is satisfied for the labeling of $u$ to $b$.  Further repetitions of the argument result in the relabeling of all interior vertices of $L$ to $b$.

\end{proof}

\section{Characterizing $\mbox{host}(K_2)$}\label{sec-k2}

Theorem~\ref{lemma-2conn-total} fully characterizes $\mbox{host}(K_2)$;
as a consequence, we can use membership in $\mbox{host}(K_2)$ as an alternate definition of $2$-connectivity.  The reconfiguration of a $2$-connected graph $G$ is achieved by defining a canonical model (one in which one vertex has one label and all other vertices have the other label) and then showing it is possible both to reconfigure any $K_2$-model to a canonical model and to reconfigure between canonical models.  In contrast, when $G$ is not $2$-connected, the presence of a cut vertex prevents reconfiguration, as no ordering of relabeling steps can prevent a branch set from being disconnected.

The proof makes use of the following observation:

\begin{observation}\label{obs-nosplit}
Let $f$ be a $K_2$-model of a graph $G$ containing a cut vertex $x$. Then at most one connected component of $G\setminus \{x\}$ contains both vertices labeled $a$ and vertices labeled $b$.
\end{observation}

\begin{proof}
Suppose two components $C_1$ and $C_2$ contain both vertices labeled $a$ and vertices labeled $b$. Since every path between a vertex in $C_1$ and a vertex in $C_2$ contains $x$, and both $G(f,a)$ and $G(f,b)$ contain vertices in $C_1$ and $C_2$, either $G(f,a)$ or $G(f,b)$ is disconnected, which contradicts the fact that $f$ is a $K_2$-model.
\end{proof}

\begin{theorem}\label{lemma-2conn-total}
$G \in \mbox{host}(K_2)$ if and only if $G$ is $2$-connected.
\end{theorem}

\begin{proof}
We begin by showing that if $G$ is $2$-connected, then $G \in \mbox{host}(K_2)$. In particular, we define a canonical model and then show that we can reconfigure from any $K_2$-model to a canonical model and between the two canonical models.  We designate one vertex $v$ as the {\em special vertex}; a canonical model is one in which $v$ has one label and all the other vertices have the other label.  Due to the $2$-connectivity of $G$, we know that each branch set is nonempty and connected and that there exists an edge between vertices in different branch sets, satisfying Definition~\ref{def-model}.  

To reconfigure any $K_2$-model $f$ to a canonical model, we reconfigure to the canonical model $g$ such that $g(v) = f(v)$ and for all $u \ne v$, $g(u) \ne f(u)$.  Without loss of generality, we assume the labels are $a$ and $b$ and let $f(v) = a$.  Suppose $f \ne g$. It follows from the $2$-connectivity of $G$ that there exists a vertex $w \neq v$ in $G(f,a)$ with an edge to $G(f,b)$. If $w$ is not a cut vertex of $G(f,a)$, then all conditions of Lemma~\ref{lemma-conds} hold for the relabeling of $w$ to $b$. If $w$ is a cut vertex, then by Lemma~\ref{lemma-leafblock} we can find a vertex $z \notin \{v,w\}$ in a leaf block of $G(f,a)$ with an edge to $G(f,b)$. Now all conditions of Lemma~\ref{lemma-conds} hold for the relabeling of $z$ to $b$. Repeated application of this argument results in the reconfiguration of $f$ to a canonical model.

To reconfigure between canonical models, it suffices to show that we can reconfigure from canonical model $g$ to non-canonical model $h$ such that $g(v) \ne h(v)$, where $v$ is the special vertex. Without loss of generality, we let $g(v) = a$ and $h(v) = b$. By Lemma~\ref{lemma-leafblock}, each leaf block of $G(g,b)$ has an edge to $G(g,a)$, and hence $v$ has a neighbour $w$ that is not a cut vertex of $G(g,b)$. Note that all conditions of Lemma~\ref{lemma-conds} hold for the relabeling of $w$ to $a$: the $2$-connectivity of $G$ implies that $|G(g,b)| \ge 2$, $w$ is not a cut vertex, $wv$ is an edge, and a model with only two branch sets cannot contain a $b$-crucial vertex. After relabeling $w$, all conditions of Lemma~\ref{lemma-conds} hold for relabeling $v$ with $b$: both $v$ and $w$ have label $a$, $v$ is not a cut vertex in the branch set, and $w$ also has a neighbour with label $b$.  We have thus provided that if $G$ is $2$-connected, $G \in \mbox{host}(K_2)$.

We now assume that $G$ is not $2$-connected, and show that $G \notin \mbox{host}(K_2)$.
Since $G$ is not $2$-connected, $G$ has a cut vertex, say $x$. Let $C_1$ and $C_2$ be any two connected components of $G\setminus \{x\}$. Consider any two $K_2$-models $f$ and $g$ such that $f(u) = a$ and $g(u) = b$ for each $u \in V(C_1)$, and $f(v) = b$ and $g(v) = a$ for each $v \in V(C_2)$. We claim $f$ is not reconfigurable to $g$.

Suppose there exists a reconfiguration sequence. Observation~\ref{obs-nosplit} implies that either all vertices in $C_1$ are relabeled to $b$ before any vertex in $C_2$ is relabeled to $a$ or vice versa. Without loss of generality, we suppose the former case. Note that in the model $f'$ obtained after all vertices in $C_1$ are relabeled to $b$, $x$ must be labeled $b$ to ensure that $G(f',b)$ is connected. Now, if $C_1$ and $C_2$ are the only components of $G\setminus \{x\}$, then there exists no vertex labeled $a$ and so $f'$ is not a $K_2$-model of $G$. Otherwise, let $C_3$ be a component that contains a vertex labeled $a$. Let $z$ be the first vertex in $C_2$ that is relabeled to $a$, which transforms a model $g'$ to $h'$. In $h'$, $x$ must be labeled $a$ to ensure $G(g',a)$ is connected. However, that means $G(g',b)$ is disconnected, and so $g'$ is not a $K_2$-model of $G$.
\end{proof}

\section{Characterizing $\mbox{host}(K_3)$}\label{sec-k3}

To show that every $3$-connected graph is in $\mbox{host}(K_3)$, we make use of Tutte's characterization in Theorem~\ref{thm-tutte-3conn}.  In order to prove Theorem~\ref{thm-3conn},  it suffices to show that wheels are in $\mbox{host}(K_3)$ (Corollary~\ref{cor-wheel}) and that connectivity is preserved under the splitting of vertices (Lemma~\ref{lemma-split}) and adding of edges (Lemma~\ref{lemma-addedge}).  

\begin{theorem}\label{thm-3conn}
Every $3$-connected graph is in $\mbox{host}(K_{3})$.
\end{theorem}

The result for wheels (Corollary~\ref{cor-wheel}) follows from a result on a generalization of wheels (Lemma~\ref{lemma-genwheel}) by allowing multiple hub vertices, each of which is a universal vertex, and replacing each rim vertex by a connected graph. We use $W(G_1,G_2,\ldots, G_m, n, \ell,m)$ to denote a generalized wheel, for each $G_i$ a connected graph on $n$ vertices, $V(G_i) = \{v_{(i,1)} \ldots v_{(i,n)}\}$, and $\ell$ and $m$ both positive integers.  The graph $W(G_1, G_2, \ldots, G_m, n, \ell,m)$ consists of $\ell$ {\em hub vertices}, $V_H = \{h_1, \ldots, h_\ell\}$, and $m n$ {\em subgraph vertices}, $V_S = \{s_{i,j} \mid 1 \le i \le m, 1 \le j \le n\}$, where $s_{i,j}$ corresponds to $v_{(i,j)}$.  The edge set consists of the {\em hub edges}, $E_H = \{h_ih_j \mid 1 \le i \le \ell, 1 \le j \le \ell, i \ne j\}$, the {\em subgraph edges}, $E_s = \{s_{i,j}s_{i,k} \mid v_{(i,j)}v_{(i,k)} \in E(G_i), 1 \le i \le m\}$, the {\em rim edges}, $E_R = \{s_{(i,j)}s_{(k,j)} \mid k \equiv i+1 \bmod{m}, 1 \le j \le n\}$,
and the {\em connecting edges}, $E_C = \{h_ks_{(i,j)} \mid 1 \le k \le \ell, 1\le i \le m, 1 \le j \le n\}$. Observe that $W(G_{1},G_{2},\ldots,G_{m},n,l,m)$ has a $K_{l+2}$-minor when $m \geq 3$. 

\begin{lemma}\label{lemma-genwheel} 
For any graphs $G_i$, $W(G_1, G_2, \ldots, G_m, n, \ell,m)$ is in $\mbox{host}(K_{\ell+2})$ for any $m \geq 3$. 
\end{lemma}

\begin{proof}

  To prove that there exists a reconfiguration sequence between any pair of $K_{\ell+2}$-models,  we identify certain models as canonical models.  It then suffices to show that we can reconfigure from any model to a canonical model and between any two canonical models. 

  We consider a $K_{\ell+2}$-model $f$ to be a {\em canonical model} if for some ordering $\pi$ on the labels $\{1, \ldots \ell+2\}$,
  $f(h_i) = \pi(i)$ for all $1 \le i \le \ell$,
$f(s_{(1,1)}) = \pi(\ell+1)$, and $f(t) = \pi(\ell + 2)$ for all $t \in V_T \setminus \{s_{(1,1)}\}$.  Clearly, this defines a $K_{\ell+2}$-minor.  For convenience, we call $s = s_{(1,1)}$ the {\em special vertex}, and at times make use of $s^+ = s_{(2,1)}$ and $s^- = s_{(m,1)}$.

We first show that any $K_{\ell+2}$-model $f$ reconfigures to a canonical model in the following steps: we ensure each hub vertex has a distinct label, then we ensure that no non-hub vertex has a label used by a hub vertex, and then we reconfigure to a canonical model. 

Suppose that not every hub vertex has a distinct label. Without loss of generality, suppose $f(h_1) = f(h_2) = b$, for some label $b$.  Because at most $\ell-1$ labels appear on hub vertices, there exists a label $a$ such that no hub vertex has label $a$.  By Lemma~\ref{lemma-universal}, we can relabel $h_2$ to $a$.  By repeating this process, we ensure that each hub vertex has a distinct label. 

Now suppose that there exists a non-hub vertex that has the same label as a hub vertex.  Because $f$ is a $K_{\ell+2}$-model, there must exist at least one non-hub vertex with a label different from a hub vertex.  Moreover, some non-hub vertex $v$ with the same label as a hub vertex must have a non-hub neighbour $w$ such that $f(w)$ is not the label of a hub vertex.  By Lemma~\ref{lemma-universal}, we can relabel $v$ with $f(w)$. By applying this idea repeatedly, we ensure that no non-hub vertex has the same label as a hub vertex.

Now each of the non-hub vertices has one of the two labels not used by any of the hub vertices.  Because the non-hub vertices form a 2-connected graph, by Theorem~\ref{lemma-2conn-total} we can reconfigure to a canonical model. 

To complete the proof, we now show that any canonical model reconfigures to any other canonical model.  As this can be accomplished by a sequences of swaps, and because Theorem~\ref{lemma-2conn-total} allows the swapping of labels of special and non-special non-hub vertices, the following two cases suffice.

\smallskip
\noindent{\bf Case 1}: Exchanging labels of two hub vertices
\smallskip

We reconfigure from canonical model $f$ to a canonical model $g$ in which for some $1 \le i \le \ell$ and $1 \le j \le \ell$, $f(h_i) = g(h_j)$, $f(h_j) = g(h_i)$, and for any other vertex $v$, $f(v) = g(v)$.  We will show that each of the following relabelings are possible: relabeling $s^+$ to $f(h_i)$, relabeling $h_i$ to $f(h_j)$, relabeling $h_j$ to $f(h_i)$, and finally relabeling $s^+$ back to $f(s^+)$.  After all four relabelings are executed, we will have reconfigured to the model $g$. 

We check each of the four proposed relabelings.  
By the definition of a canonical model, we know that $f(s^+) \ne f(s)$, that $s^+$ has at least one neighbour with each possible label, and that $s^+$ is not a cut vertex of $G(f,f(s^+))$.
We verify that all the conditions of Lemma~\ref{lemma-conds} are satisfied by the relabeling of $s^+$ to $f(h_i)$:  there are other vertices with label $f(s^+)$, $s^+$ is not a cut vertex, $s^-$ has the same neighbours in other branch sets, and  $h_i$ is a neighbour of $s^+$.  To see that we can relabel $h_i$ to $f(h_j)$, we observe that $G(f,f(h_i))$ contains both $h_i$ and $s^+$, that $h_i$ is not a cut vertex of the branch set, that $s^+$ has a neighbour with each label, and that $h_i$ is a neighbour of $h_j$; thus, all conditions of Lemma~\ref{lemma-conds} are satisfied. 
Finally, relabeling $h_j$ to $f(h_i)$ and then from $s^+$ to $f(s^+)$ both follow
from Lemma~\ref{lemma-universal} as in each case there is a universal vertex with the same label.  

\smallskip

\noindent{\bf Case 2}: Exchanging labels of a hub vertex and the special vertex

\smallskip

In this case, the canonical model $f$ is reconfigured to the canonical model $g$, where for some hub vertex $h_i$, $f(h_i) = g(s)$, $f(s) = g(h_i)$, and for all other vertices, $f(v) = g(v)$.  By arguments similar to those given in the previous case, we can show that we can label $s^+$ by $f(h_i)$, $h_i$ by $f(s)$, $s$ by $f(h_i)$, and finally $s^+$ by $f(s^+)$.

\end{proof}

\begin{corollary}\label{cor-wheel}
$W_{n}$ is in $\mbox{host}(K_{3})$.
\end{corollary}

We consider the splitting of vertices in two steps.  In Lemma~\ref{lemma-split-samelabel}, which applies more generally to $K_k$ for any $k > 2$, we show that we can reconfigure between models in which the vertices resulting from the split have the same label. Then, in Lemma~\ref{lemma-split}, which uses Lemma~\ref{lemma-slurp-siphon}, we consider cases in which the vertices can have different labels.

\begin{lemma}\label{lemma-split-samelabel}
Let $G$ be a $2$-connected graph and $G'$ be formed from $G$ by splitting a vertex $v$ into vertices $x$ and $y$. For any $k > 2$, let $f$ and $g$ be $K_k$-models of $G$ and $f'$ and $g'$ be $K_k$ models of $G'$ such that $f(v) = f'(x) = f'(y)$, $g(v) = g'(x) = g'(y)$ and for all  $u \in V(G) \setminus \{v\}$, $f(u) = f'(u)$ and $g(u) = g'(u)$. If $f$ and $g$ are reconfigurable, then $f'$ and $g'$ are reconfigurable.
\end{lemma}

\begin{proof}
Since $f$ and $g$ are reconfigurable, there is a reconfiguration sequence $\sigma = f=f_{1},\ldots,f_{\ell}= g$ for some value of $\ell$. Using $\sigma$, we wish to form a reconfiguration sequence $\sigma'$ from $f'$ to $g'$ in the reconfiguration graph for $G'$.  In forming the sequence, we observe that if there is a prefix $\tau$ of $\sigma$ such that $v$ is not relabeled in any of the steps, then we can form a prefix $\tau'$ of $\sigma'$ by executing the same sequence of steps.  

We now consider the first relabeling of $v$ in $\sigma$, say from $f_{j}$ to $f_{j+1}$; we wish to show that in $\sigma'$, we can relabel both $x$ and $y$ in the same way.  Without loss of generality, we assume that $f_{j}(v) = f'_{j}(x) = f'_{j}(y) = a$ and that $f_{j+1}(v) = b$.

We can use the fact that Lemma~\ref{lemma-conds} holds for the labeling of $v$ by $b$ to establish useful properties of $x$ and $y$.  Because $v$ is not in a branch set of size one (condition~\ref{cond-nonempty}), $x$ and $y$ are not the only two vertices in $G'(f'_{j},a)$. Since $v$ is a not a cut vertex of its branch set (condition~\ref{cond-notcut}), $x$ and $y$ together cannot form a cut set.  As $v$ has a neighbour with label $b$ (condition~\ref{cond-nbr}), either $x$ or $y$ (or both) must have a neighbour with label $b$ under $f'_{j}$.  Finally, because $v$ is not a $b$-crucial vertex (condition~\ref{cond-edges}), there must exist some vertex other than $x$ or $y$ in $G'(f'_{j},a)$ that has a neighbour with label $b$.

Without loss of generality, we assume that $x$ has a neighbour with label $b$, and consider two cases, depending on whether or not $x$ is a joining vertex in the block tree of $G'(f'_{j},a)$.

\smallskip

\noindent{\bf Case 1}: $x$ is a not a joining vertex of the block tree of $G'(f'_j, a)$

\smallskip

By our observations above, all conditions of Lemma~\ref{lemma-conds} hold for the relabeling of $x$ to $b$.  As a consequence of the relabeling, $y$ now has a neighbour with label $b$.  Because $v$ was not a cut vertex in $G(f_j,a)$, the removal of both $x$ and $y$ cannot disconnect $G'(f'_j,a)$, and hence condition~\ref{cond-notcut} holds for $y$.  As the remaining conditions of Lemma~\ref{lemma-conds} were established in the argument above, we can now relabel $y$ to $b$, as needed. 

\smallskip
  
\noindent{\bf Case 2}: $x$ is a joining vertex in the block tree of $G'(f'_j,a)$

\smallskip

We first show that the component $C$ of $G'(f'_j,a)$ containing $y$ consists solely of the vertex $y$.  If instead there existed another vertex in $C$, then the removal of both $x$ and $y$ would separate the vertex from the other components of $G'(f'_j,a)\setminus\{x\}$, and consequently $v$ would be a cut vertex in $G(f_j,a)$, which contradicts Lemma~\ref{lemma-conds} for the relabeling of $v$ to $b$.

By the definition of the split operation, $\deg(y) \geq 3$; because $x$ is $y$'s only neighbour with label $a$, $y$ must have at least one neighbour with label $b$ or $c$, for some $c \notin \{a,b\}$.  If $y$ has a neighbour with label $b$, we can use Case 1 to complete the relabeling of $y$ and then $x$.  Otherwise, we will show that we can first relabel $y$ to $c$, relabel $x$ to $b$, and finally relabel $y$ to $b$.  In each case, we show that we can use Lemma~\ref{lemma-conds}.

To see that we can relabel $y$ to $c$, it suffices to observe that $y$ is not $c$-crucial, as $x$ has a neighbour in $b$, and if $y$ were essential for some $d \notin \{b,c\}$, then $v$ would be essential for $d$ in $f_j$, and hence $b$-crucial in $f_j$, which is a contradiction.  Now, since $x$ is now no longer a cut vertex, we can relabel $x$ to $b$. Finally, since $y$ now has a neighbour with label $b$ (that is, $x$), we can relabel $y$ with $b$, as needed. 
\end{proof}

\begin{lemma}\label{lemma-slurp-siphon}
Given a $3$-connected graph $G$ and a $K_3$-model $f$ of $G$, suppose there exists a cut vertex $x$ of $G(f,a)$ with a neighbour in $G(f,b)$ such that $x$ is not essential for $c$.  Then there exists a component $D$ of $G(f,a)\setminus\{x\}$ such that we can reconfigure $f$ to a model $g$ in which $g(v)=f(v)$ for each $v\in V(D)$, $g(x)=b$, and $g(u)\neq f(u)$ for all other vertices of $G(f,a)$.
\end{lemma}  

\begin{proof}
We form model $g$ by first siphoning all components of $G(f,a)\setminus\{x\}$ except $D$ out of $G(f,a)$ and then by relabeling $x$ to $b$. 

Because $x$ is not essential for $c$, there must exist at least one component of $G(f,a) \setminus \{x\}$ that contains a vertex with a neighbour in $G(f,c)$; we choose $D$ to be one such component.

By Lemma~\ref{lemma-leafblock}, each other component $C$ must have a neighbour in either $G(f,b)$ or $G(f,c)$.   Since $x$ has a neighbour in $G(f,b)$ and $D$ has a neighbour in $G(f,c)$, $f$ does not hit a leaf-crucial model,  a leaf-$b$-crucial model,  or a leaf-$c$-crucial model on relabeling $C$, and hence we can use Lemma~\ref{lemma-slurp-component} to relabel all vertices of $C$ with either $b$ or $c$. 

After the vertices of every component except $D$ have been relabeled, all the conditions for Lemma~\ref{lemma-conds} now hold for the relabeling of $x$ to $b$: the branch set for label $a$ has at least one vertex other than $x$, $x$ is no longer a cut vertex, $x$ has a neighbour labeled $b$, and a vertex in $D$ has a neighbour labeled $c$. 

\end{proof}

\begin{lemma}\label{lemma-split}
  Suppose $G$ is a $3$-connected graph such that $G \in \mbox{host}(K_{3})$ and $G'$ is a graph formed from $G$ by splitting
a vertex $v$ into vertices $x$ and $y$. Then $G'$ is in $\mbox{host}(K_{3})$.
\end{lemma}

\begin{proof}
  We show that for any source and target $K_3$-models of $G'$, we can find a reconfiguration sequence from the source to the target.  We know from Lemma~\ref{lemma-split-samelabel} that we can reconfigure between any two $K_3$-models in which $x$ and $y$ have the same labels. Here, we show that we can reconfigure any $K_3$-model to a $K_3$-model in which $x$ and $y$ have the same labels. This suffices to demonstrate the existence of a reconfiguration sequence between the source and target $K_3$-models, as we reconfigure from the source $K_3$-model to a $K_3$-model in which $x$ and $y$ have the same labels, then to another $K_3$-model in which $x$ and $y$ have the same labels, and finally to the target $K_3$-model.

  Without loss of generality, we assume that the labels are $a$, $b$, and $c$, and that in the starting $K_3$-model, $f'(x) = a$ and $f'(y) = b$. If Lemma~\ref{lemma-conds} holds for either relabeling $x$ to $b$ or relabeling $y$ to $a$, then we can accomplish the reconfiguration in a single step.  Similarly, the reconfiguration can be accomplished using Lemma~\ref{lemma-slurp-siphon} if either $x$ or $y$ is a cut vertex that is not essential for $c$. 

    Thus, it suffices to consider the cases in which either condition~\ref{cond-nonempty} or \ref{cond-edges} of Lemma~\ref{lemma-conds} must be violated for both relabeling $x$ to $b$ and relabeling $y$ to $a$.  In any case, both $x$ and $y$ will be essential for $c$.  By Lemma~\ref{lemma-all-weak}, it is not possible for both $x$ and $y$ to be essential for $c$ unless $|G'(f',a)| = |G'(f',b)| = 1$.  Thus, it suffices to consider the case $|G'(f',a)| = |G'(f',b)| = 1$. 

    Due to the $3$-connectivity of $G$, $x$ and $y$ will each have at least two neighbours in $G'(f',c)$.  We will show that one of $y$'s neighbours $w$ can be relabeled $b$, after which $y$ can be relabeled $a$.

    If Lemma~\ref{lemma-conds} does not apply for the relabeling of $w$ to $b$, then because $w$ is not $b$-crucial, the only possible condition of Lemma~\ref{lemma-conds} that can be violated is condition~\ref{cond-notcut}.  Suppose that every neighbour of $y$ in $G'(f',c)$ is a cut vertex. Since by Lemma~\ref{lemma-leafblock}, each leaf block in $G'(f',c)$ has two interior vertices that are endpoints of connecting edges, each of these edges must connect to $x$.
By $3$-connectivity, there must be three vertex-disjoint paths to $y$ from an interior vertex in a leaf block in $G(f',c)$.  As only one can pass through $x$ and only one can pass through the joining vertex of the leaf block, there can only be two paths at most,  forming a contradiction.  Because all conditions of Lemma~\ref{lemma-conds} must hold, we can relabel $w$ to $b$. 

  To see that we can now relabel $y$ to $a$, we observe that since $G'(f',c)$ was connected, $z$ has a neighbour with label $c$. This implies that $y$ is not essential for $c$, as needed to satisfy all conditions of Lemma~\ref{lemma-conds}. 
\end{proof}

Finally, in Lemma~\ref{lemma-addedge}, we show that for $G'$ the graph formed by adding an edge $xy$ to a $3$-connected graph $G \in \mbox{host}(K_3)$, $G'$ is also in $\mbox{host}(K_3)$.  We achieve the result by showing that we can handle situations in which $xy$ plays a role not played by any other edge, either as an essential edge or as a bridge within a branch set.  The proof relies on the following results:

\begin{lemma}\label{obs-cor-ess-vertex}
Given a $2$-connected graph $G$ and a $K_3$-model $f$ of $G$, any branch set containing at least two vertices has no crucial vertex.  
\end{lemma}

\begin{proof}
We will show that for a $k$-connected graph $G$ and an $H$-model of $G$ such that $|V(H)|$ is odd and $k>(|V(H)|-1)/2$, any branch set containing at least $k$ vertices has at most $(|V(H)|-3)/2$ crucial vertices.  The result then follows from the fact that $(|V(K_3)|-3)/2=0$.

We first observe that any branch set $G(f,a)$ contains at most $(|V(H)|-1)/2$ crucial vertices, because each crucial vertex must be essential for at least two distinct labels in the set $V(H) \setminus \{a\}$.  To prove the lemma, it thus suffices to consider the case in which $G(f,a)$ contains at least $k$ vertices and has exactly $\lfloor (|V(H)|-1)/2 \rfloor$ crucial vertices.  When $G(f,a)$ has $\lfloor (|V(H)|-1)/2 \rfloor$ crucial vertices, each crucial vertex in $G(f,a)$ has exactly two neighbours in different branch sets.  Moreover, because together all the crucial vertices are essential for all the labels in $V(H) \setminus \{a\}$, no non-crucial vertex in $G(f,a)$ has a neighbour in a different branch set.  Thus, the crucial vertices form a cut set which separates the non-crucial vertices of $G(f,a)$ from the rest of the vertices of $G$, which is a contradiction since $G$ is $k$-connected.  Hence, $G(f,a)$ contains at most $(|V(H)|-3)/2$ crucial vertices. 
\end{proof}

\begin{lemma}\label{lemma-ben}
  Given a $3$-connected graph $G$ and a $K_{3}$-model $f$ of $G$, suppose that $x \in G(f,a)$, $y \in G(f,b)$, and $xy$ is an essential edge. Then we can reconfigure $f$ to a $K_{3}$-model $g$ such that $g(v) \ne f(v)$ for each $v \in G(f,c)$, $c \notin \{a, b\}$, $g(u) = f(u)$ for all other vertices $u$, and $xy$ is not an essential edge in $g$.
\end{lemma}

\begin{proof}
We consider two cases, depending on whether or not $G(f,c)$ is $2$-connected.  
  \smallskip

\noindent{\bf Case 1}: $G(f,c)$ is not $2$-connected.
\smallskip

By Lemma~\ref{lemma-leafblock}, each leaf block $L$ in $G(f,c)$ has at least two interior vertices that are endpoints of connecting edges.  Due to $3$-connectivity, we can further show that each leaf block in $G(f,c)$ has edges to both $G(f,a)$ and $G(f,b)$, as otherwise the joining vertex of the leaf block and either $x$ or $y$ would form a cut set of size two separating internal vertices of the leaf block at one of the branch sets.

The fact that all other leaf blocks connect to both $G(f,a)$ and $G(f,b)$ ensure that
$f$ does not hit a leaf-crucial model on relabeling $L$. We can then use Lemma~\ref{lemma-relabel-leafblock} to relabel all interior vertices of $L$. Due to the connectivity of $L$ and the fact that it contained neighbours in both $G(f,a)$ and $G(f,b)$, it follows that $xy$ is not an essential edge in the resulting model $g$.

\smallskip
\noindent{\bf Case 2}: $G(f,c)$ is $2$-connected. 
\smallskip

We first use $3$-connectivity to show that $|G(f,c)| > 1$ and that
$G(f,c)$ contains vertices $u$ and $v$ such that $u$ has a neighbour
in $G(f,a)$ and $v$ has a neighbour in $G(f,b)$.  If $|G(f,c)| = 1$,
then the vertex in $G(f,c)$ and either $x$ or $y$ form a cut set of
size two separating the branch sets $G(f,a)$ and $G(f,b)$.  Similarly, if $G(f,c)$ contained only a single endpoint of a connecting edge, then the endpoint and either $x$ or $y$ would also form a cut set of size two.

We can choose $u$ and $v$ such that $P$ is a $(u,v)$-path in $G(f,c)$ such that no vertex in $P$ other than $u$ or $v$ has a neighbour in $G(f,a)$ or $G(f,b)$. 
We let $u =v_{1},\ldots,v_{t} =v$ be the vertices of $P$ with edges $v_{i}v_{i+1}$, $i \in \{1,\ldots,t\}$.   Since $G$ is $3$-connected and $xy$ is an essential edge, $G(f,c)$ must contain vertices $w \notin \{u,v\}$ and $z \notin \{u,v\}$ such that $w$ has an edge to $G(f,a)$ and $z$ has an edge to $G(f,b)$.

We will attempt to relabel all of $P$ to label $a$. As $G(f,c)$ is $2$-connected, we can relabel $v_{1}$ to $a$. If the resulting branch set is $2$-connected, then we attempt relabel $v_{2}$,$v_{3}$, $\ldots, v_{t}$ in that order until we relabel the entire path. If this succeeds, then we are done. Otherwise, at some step relabeling along the path we obtain a $K_{3}$ model $g$ such that $G(g,c)$ is not $2$-connected. Now we apply Case 1 to $g$ to complete the claim.  
\end{proof}

\begin{lemma}\label{lemma-addedge}
Suppose $G$ is a $3$-connected graph such that $G \in \mbox{host}(K_{3})$ and $G'$ is formed from $G$ adding an edge $xy$. Then $G' \in \mbox{host}(K_{3})$.
\end{lemma}

\begin{proof}
We first observe that any $K_3$-model of $G$ is also $K_3$-model of $G'$.  Consequently, to show that we can reconfigure between any $K_3$-models of $G'$, it suffices to show that we can reconfigure between any $K_3$-model of $G'$ and a $K_3$-model of $G$, as the fact that $G \in \mbox{host}(K_3)$ ensures that we can reconfigure between any two $K_3$-models of $G$.
  
There are only two cases in which a $K_3$-model $f$ of $G'$ is not a $K_3$-model of $G$, namely cases in which the role $xy$ plays in the $K_3$-model is not played by any other edge.  In both cases we can assume that $G' \ne G$, and consequently that $|V(G')| > 3$.

  \smallskip
  \noindent{\bf Case 1}: $xy$ is the essential edge connecting $G'(f,f(x))$ and $G'(f,f(y))$
  \smallskip

Without loss of generality, we assume $f(x) = a$ and $f(y) = b$. By Lemma~\ref{lemma-ben}, we can reconfigure $f$ to a $K_3$-model $f'$ by relabeling only vertices in $G(f,c)$ such that $xy$ is not an essential edge in $f'$, which means $f'$ is also a $K_3$-model of $G$.

  \smallskip
  \noindent{\bf Case 2}: $f(x) = f(y)$ and $xy$ is a bridge in $G'(f,f(x))$
  \smallskip

We show that we can reconfigure to a model in which $x$ and $y$ have different labels so that $xy$ is a connecting edge. Depending on whether $xy$ is then an essential edge, we have either completed the reconfiguration or we have reduced the situation to Case 1.
  
Without loss of generality, we assume that $f(x) = f(y) = a$, and observe that the removal of $xy$ separates $G'(f,a)$ into two components $C_1$ (containing $x$) and $C_2$ (containing $y$), each of which contains at least one leaf block.

By Lemma~\ref{lemma-leafblock}, each of the leaf blocks has two vertices with neighbours in $G'(f,b)$ or $G'(f,c)$.  We will show that we can reconfigure to a $K_3$-model in which either $C_1$ or $C_2$ has no vertex with label $a$, so that $xy$ is no longer a bridge.

    \smallskip
    \noindent{\bf Case 2a}:  One component has edges to both $G'(f,b)$ and $G'(f,c)$.
    \smallskip

    Without loss of generality, let $C_1$ have connecting edges to both $G'(f,b)$ and $G'(f,c)$. Then $f$ does not hit a leaf-$b$-crucial model or a leaf-$c$-crucial model on relabeling $C_2$ because there are necessary connecting edges from $C_1$. Also, $f$ does not hit a leaf-crucial model on relabeling $C_2$ because there are at least two vertices labeled $a$ in each model in the reconfiguration sequence, and so by Lemma~\ref{obs-cor-ess-vertex}, there never exists a crucial vertex labeled $a$. Now by Lemma~\ref{lemma-slurp-component}, since $x$ is a cut vertex of $G'(f,a)$, we can relabel all the vertices of $C_2$, which ensures $xy$ is a connecting edge.
    
    \smallskip
    \noindent{\bf Case 2b}: One component has two edges to $G'(f,b)$ and one component has two edges to $G'(f,c)$.
    \smallskip

Without loss of generality, let $C_1$ have connecting edges to $G'(f,b)$ and $C_2$ have connecting edges to $G'(f,c)$. Then $f$ does not hit a leaf-$c$-crucial model on relabeling $C_2$ because $C_1$ has edges to $G'(f,b)$. Also, it follows from Lemma~\ref{obs-cor-ess-vertex} that $f$ does not hit a leaf-crucial model on relabeling $C_2$ because there are at least two vertices labeled $a$ in each model in the reconfiguration sequence. Hence, by Lemma~\ref{lemma-slurp-component}, since $x$ is a cut vertex of $G'(f,a)$, we can relabel all the vertices of $C_2$, which ensures $xy$ is a connecting edge.

\end{proof}

\section{Characterizing $\mbox{host}(K_4)$}\label{sec-k4}

In order to use Ding and Qin's characterization in Theorem~\ref{thm-dingqin}, we show that $C^2_6 \in \mbox{host}(K_4)$ (Lemma~\ref{lemma-base-4conn}) and $K_5 \in \mbox{host}(K_4)$ (a special case of Lemma~\ref{lemma-clique}) and present analogues of Lemmas~\ref{lemma-split} and \ref{lemma-addedge} (Lemmas~\ref{lemma-split-4} and \ref{lemma-addedge-4}), showing that $\mbox{host}(K_{4})$ is closed under the splitting of vertices or adding of edges for $4$-connected graphs.  The four results are sufficient to prove the following theorem:

\begin{theorem}\label{thm-4conn}
  Every $4$-connected graph is in $\mbox{host}(K_{4})$, provided it is not in $\mathcal{L}$, where $\mathcal{L} = \{H : H$ is the line graph of an internally $4$-connected cubic graph$\}$. 
\end{theorem}

The results establishing the base cases of the characterization are relatively straightforward.  Lemma~\ref{lemma-base-4conn} makes use of various properties of the structure of $C_{6}^{2}$, most notably the fact that for each vertex $v$, there is a unique vertex $s(v)$ such that there is no edge between $v$ and $s(v)$.  As an immediate consequence, any four vertices form a $4$-cycle, which permits the use of Theorem~\ref{lemma-2conn-total} for reconfiguration of part of the graph.  Lemma~\ref{lemma-clique}, a generalization of $K_5 \in \mbox{host}(K_4)$, follows easily from the high connectivity of cliques.

\begin{lemma}\label{lemma-base-4conn}
$C_{6}^{2}$ is in $\mbox{host}(K_{4})$. 
\end{lemma}

\begin{proof}
  We first observe that every $K_{4}$-model of $C_{6}^{2}$ either has one branch set of size three (type A) or two branch sets of size two (type B).  Moreover, we can easily reconfigure from a $K_4$-model of type A to one of type B, as the branch sets that are not of size three are each of size one and form a $K_3$. Therefore it suffices to show that we can reconfigure any type B $K_4$-model to any other $K_4$-model.

  We can form a reconfiguration sequence using transformations of the three types outlined below, where in each case we use $a$ and $b$ as the labels of the two branch sets of size two and $c$ and $d$ as the labels of the two remaining branch sets.  In Case 1, we consider rearrangements of the vertices with labels $a$ and $b$ and in Case 2, the exchanging of labels for the two branch sets of size one.  The remaining step is covered in Case 3, in which there is an exchange of labels of one vertex with label $a$ and one with label $d$.
  
  We make extensive use of the fact that any assignments of labels to vertices forms a $K_4$-model provided that all four labels are used and there are no branch sets of size one consisting of $v$ and $s(v)$, for any vertex $v$.  In particular, there is always a connecting edge between branch sets when at least one branch set is of size greater than one.

  \smallskip
  \noindent{\bf Case 1:}  $G(f,c) = G(g,c)$ and $G(f,d) = G(g,d)$
  \smallskip

  We observe that any collection of four vertices from $C_{6}^{2}$ lie on a $4$-cycle, and that by Theorem~\ref{lemma-2conn-total}, $C_{4}$ is in $\mbox{host}(K_{2})$.  Consequently, we can reconfigure $f$ to $g$ without changing the vertices with labels $c$ and $d$.

  \smallskip
  \noindent{\bf Case 2:} $G(f,c) = G(g,d)$ and $G(f,d) = G(g,c)$
  \smallskip

  We execute the relabeling in a series of steps, ensuring that there are never two branch sets of size one containing $v$ and $s(v)$ for any vertex $v$, as is sufficient to guarantee a $K_4$-model. Since each vertex in the graph has degree four, we can assume without loss of generality that the vertex $u$ in $G(f,d)$ is adjacent to both vertices in $G(f,a)$ and that the vertex $v$ in $G(f,c)$ is adjacent to both vertices in $G(f,b)$.  

  First, we choose a vertex $w$ in $G(f,a)$ to relabel to $d$, ensuring that the remaining vertex $x$ is not $s(v)$.  Next, we choose a vertex $y$ in $G(f,b)$ to relabel to $c$, ensuring that the remaining vertex $z$ is not $s(x)$.

  At this point we have a $K_4$-model $j$ such that $G(j,a) = \{x\}$, $G(j,b) = \{z\}$, $G(j,c) = \{v, y\}$, and $G(j,d) = \{u, w\}$.  By the argument given in Case 1, we can now reconfigure to a $K_4$-model $k$ such that $G(k,a) = \{x\}$, $G(k,b) = \{z\}$, $G(k,c) = \{u, w\}$, and $G(k,d) = \{v, y\}$.

  We can now safely relabel $y$ by $b$, as $x$ is not $s(v)$, and then relabel $w$ by $a$, as there must be an edge between $u$ and $v$ (namely the connecting edge between between $G(f,c)$ and $G(f,d)$).  

  \smallskip
  \noindent{\bf Case 3:} $G(f,a) = \{w, x\}$, $G(f,d) = \{u\}$, $G(g,a) = \{u, w \}$, $G(g,d) = \{x\}$
  \smallskip

  We observe that there must be edges $uv$ and $vx$, as connecting edges between branch sets in models $f$ and $g$.  In addition, there must be edges $wx$ and $yz$, for $G(f,b) = \{y, z\}$, as the vertices in each branch set must be connected.  

  Using the technique in Case 2, we first reconfigure the two branch sets of size two so that $x$ retains label $a$ and $s(v)$ has label $b$ (where no change is required if $w \ne s(v)$).  We can now relabel $x$ to $d$, as the two branch sets of size one contain $v$ and a vertex that is not $s(v)$, and then $u$ to $a$.  Now the technique from Case 2 can be used to relabel the vertices with labels $a$ and $b$, as needed. 
\end{proof}

\begin{lemma}\label{lemma-clique} For any $m > \ell$,
$K_{m} \in \mbox{host}(K_\ell)$. 
\end{lemma}

\begin{proof}

  To show that it is possible to reconfigure between any two $K_\ell$-models, it suffices to show that any $K_\ell$-model can be reconfigured to the same canonical model.  We denote $V(K_m) = \{v_1, \ldots, v_m\}$ and $V(K_\ell) = \{u_1, \ldots, u_\ell\}$, and the canonical $K_\ell$-model $g$ such that for each $i \in \{1,\ldots,\ell\}$, $g(v_{i}) = u_i$, and for each $i \in \{\ell + 1,\ldots,m\}$, $g(v_{i}) =u_\ell$.

 We now consider an arbitrary $K_{\ell}$-model $f$, and show that it can be reconfigured to the canonical model.  We first relabel each vertex in  $\{v_{\ell}+1,\ldots,v_{m}\}$ with the label $u_\ell$.  Since $K_m$ is a complete graph, the only condition in which a vertex cannot be relabeled is if it is the only vertex in its branch set.  Suppose there is such a vertex $v_i$, $i \in \{\ell+1, \ldots,m\}$.  In this case, there must be two vertices $v_j$ and $v_k$ with the same label, one of which can be relabeled to $f(v_i)$.  Then $v_i$ can be relabeled, as desired.  By repeating this process, we ensure that for all $i \in \{\ell+1, \ldots, m\}$, $v_i$ has the label $u_\ell$.

  To complete the relabeling, it suffices to show that we can swap the labels of any vertices $v_i$ and $v_j$ for $i,j \in \{1, \ldots, \ell\}$, ensuring that by the end of all the swaps, $v_i$ has the label $u_i$.  This can easily be achieved by choosing any vertex $v_k$, $k \in \{\ell+1, \ldots, m\}$ and executing the following sequence of relabelings, starting with model $f$: $v_k$ is relabeled $f(v_i)$, $v_i$ is relabeled $f(v_j)$, $v_j$ is relabeled $f(v_i)$, and then $v_k$ is relabeled $u_\ell$.  At each step, there is at least one vertex with each label, as required.  
\end{proof}

As essential edges can result from either the splitting of vertices or the
adding of edges, Lemma~\ref{lemma-xyess} plays a crucial role in the proofs of both Lemma~\ref{lemma-split-4} and Lemma~\ref{lemma-addedge-4}.  The proof of Lemma~\ref{lemma-xyess} makes extensive use of Lemmas~\ref{lemma-all-weak}, \ref{lemma-four-weak}, and \ref{lemma-minislurp} in covering all possible cases of weak connections among branch sets, where branch sets of size one are handled separately.

\begin{lemma}\label{lemma-four-weak}
  Given a $4$-connected graph $G$ and a $K_4$-model $f$ of $G$ such that for labels $a$, $b$, $c$, $d$ there exist weak connections between branch sets with labels $a$ and $b$, $b$ and $c$, $c$ and $d$, and $d$ and $a$, then it is not possible to designate lynchpins such that the branch set with label $a$ contains at loeast one vertex $x$ that is not a lynchpin and the branch set with label $d$ contains at least one vertex $y$ that is not a lynchpin.
\end{lemma}

\begin{proof}
  As in the proof of Lemma~\ref{lemma-all-weak}, we demonstrate that if it is possible to designate lynchpins in the way described in the statement of the lemma, then we violate the $4$-connectivity of $G$ by finding a cut set of size at most three that separates $x$ and $y$.

 Each path between $x$ and $y$ must pass through the lynchpin in the weak connection between $G(f,a)$ and $G(f,d)$, the lynchpin in the weak connection between $G(f,a)$ and $G(f,b)$ (if the path goes through the branch sets with labels $a$, $b$, and $d$ or branch sets with labels $a$, $b$, $c$, $d$, in those orders), the lynchpin in the weak connection between $G(f,b)$ and $G(f,c)$ (if the path goes through branch sets with labels $a$, $b$, $c$, and $d$, in that order), or the lynchpin in the weak connection between $G(f,c)$ and $G(f,d)$ (if the path goes through the branch sets with labels $a$, $c$, and $d$, or branch sets with labels $a$, $b$, $c$, $d$, in those orders).  In each case one of the at most three lynchpins is on the path, forming a cut set of size at most three separating $x$ and $y$. 
\end{proof}

\begin{lemma}\label{lemma-minislurp}
Given a $3$-connected graph $G$ and an $K_4$-model $f$ of $G$, suppose $xy$ is an essential edge and there exists an essential edge $e$ from $G(f,f(y))$ to $G(f,f(z))$, $z \ne x$, such that $y$ is not the endpoint of $e$.  
Then it is possible to reconfigure $f$ to a model $g$ in which $g(x) = g(y) = f(x)$, and for all $a \ne f(y)$, for $v \in G(f,a)$, $g(v) = f(v)$.
\end{lemma}

\begin{proof}
  Without loss of generality, we let $f(x)=a$, $f(y)=b$, and $f(z) = c$, so that there is an essential edge $\mbox{ess}(b,c)$. We consider two cases, depending on whether or not $y$ is a cut vertex.
  
  When $y$ is not a cut vertex, we verify that all conditions of Lemma~\ref{lemma-conds} hold for relabeling $y$ to $a$: condition~\ref{cond-nonempty} follows as $G(f,b)$ contains at least $y$ and an endpoint $v$ of $\mbox{ess}(b,c)$, condition~\ref{cond-notcut} follows by assumption, condition~\ref{cond-nbr} follows from the existence of $xy$, and condition~\ref{cond-edges} follows from the observation that if $y$ is an $a$-crucial vertex, then it must be essential for $d$, which implies $\{y,z\}$ is a $2$-cut in $G$, contradicting the fact that $G$ is 3-connected.

  If instead $y$ is a cut vertex, by removing $y$ we can break $T(G,f,b)$ into components such that one of the components $C$ contains the endpoint of $\mbox{ess}(b,c)$.   By Lemma~\ref{lemma-leafblock}, each leaf block of $C$ must contain at least two vertices with neighbours in other branch sets, and hence $C$ must contain an edge with a neighbour in $G(f,d)$.  As $C$ has all necessary connecting edges, we can apply Lemma~\ref{lemma-slurp-component} to siphon away the vertices in every other component $C' \neq C$.  Consequently, $y$ will no longer be a cut vertex, we can then relabel $y$ to $a$, as needed.  
\end{proof}

\begin{lemma}\label{lemma-xyess}
Suppose $G$ is a $4$-connected graph such that $G \in \mbox{host}(K_{4})$, $|V(G)| \ge 5$, and $xy$ is an essential edge under the $K_4$-model $f$.  Then it is possible to reconfigure $G$ to a $K_4$-model in which $x$ and $y$ have the same label.
\end{lemma}

\begin{proof}

  Without loss of generality, we assume that $f(x) = a$, $f(y) = b$, and that the two remaining labels are $c$ and $d$.    As the relabeling of $x$ or $y$ can be achieved in a single step if Lemma~\ref{lemma-conds} holds, in the remainder of the proof we assume instead that the lemma holds for neither $x$ nor $y$.  Similarly, since the result follows immediately from Lemma~\ref{lemma-minislurp}, we consider only cases in which the lemma does not apply.  That is, we assume that there is no essential edge with an endpoint that is either in $G(f,a) \setminus \{x\}$ or in $G(f,b) \setminus \{y\}$.  

  We consider all possibilities for the sizes of branch sets and weak connections among them.  Due to symmetry between $a$ and $b$ and symmetry between $c$ and $d$,  the cases listed handle all possible situations.  Cases 1 and 2 handle the situations in which $|G(f,a)| = |G(f,b)| = 1$, Cases 3--6 consider situations in which at most one of $G(f,a)$ and $G(f,b)$ can consist of a single vertex, and the remaining cases consider situations in which $G(f,a)$ and $G(f,b)$ each consist of more than a single vertex.

  \smallskip
  \noindent {\bf Case 1}: $|G(f,a)| = |G(f,b)| = |G(f,c)| = 1$
  \smallskip

  We use $z$ to denote the vertex in $G(f,c)$ and observe that by $4$-connectivity, each of $x$, $y$, and $z$ have at least two neighbours in $G(f,d)$. In order to relabel $y$ to $a$, we need to first fill $G(f,b)$ with vertices from $G(f,d)$ such that at least one has a neighbour in $G(f,c)$.  We omit the symmetrical case in which we could fill $G(f,b)$ with vertices from $G(f,d)$ in order to relabel $x$ to $b$.

  \smallskip
    \noindent{\bf Case 1a}: $y$ has a neighbour $w$ in $G(f,d)$ such that $w$ is a cut vertex of $G(f,d)$
  \smallskip

We observe that as Lemma~\ref{lemma-two-weak}, point one holds for each pair of branch sets in $G(f,a)$, $G(f,b)$, and $G(f,c)$, each leaf block in $G(f,d)$ contains interior vertices that are endpoints of connecting edges to each of $G(f,a)$, $G(f,b)$, and $G(f,c)$. We can then select a component $C$ of $G(f,d) \setminus \{w\}$ to retain in $G(f,d)$, using Lemma~\ref{lemma-slurp-component} to siphon away all other components.  Now $w$ is no longer a cut vertex; it can be relabeled to $b$ and $y$ to $a$.

  \smallskip
  \noindent{\bf Case 1b}: $y$ has no neighbour in $G(f,d)$ that is a cut vertex of $G(f,d)$
  \smallskip

  If $y$ has a neighbour $w$ in $G(f,d)$ such that $w$ is a neighbour of a vertex in $G(f,c)$, then we can relabel $w$ to $b$ in a single step using Lemma~\ref{lemma-conds}, and in another step we can then relabel $y$ to $a$.
  
   Otherwise, because Lemma~\ref{lemma-conds} applies to the relabeling of neighbour of $y$ to $b$, we can iteratively relabel a neighbour, a neighbour of a neighbour, and so on until we encounter either a vertex that is a neighbour of a vertex in  $G(f,c)$, which can be relabeled to $b$ using Lemma~\ref{lemma-conds}, or a cut vertex, to which Case 1a applies.

   \smallskip
   \noindent {\bf Case 2}: $|G(f,a)| = |G(f,b)| = 1$, $|G(f,c)| > 1$, and $|G(f,d)| > 1$.  
   \smallskip

   We first show that there cannot be a weak connection between $|G(f,c)|$ and $|G(f,d)|$.
  By Lemma~\ref{lemma-four-weak}, we know it is not possible to designate lynchpins such that both $G(f,c)$ and $G(f,d)$ contain vertices that are not lynchpins of weak connections between $G(f,b)$ and $G(f,c)$, $G(f,c)$ and $G(f,d)$, and $G(f,d)$ and $G(f,a)$.  Since at most three vertices in $G(f,c)$ and $G(f,d)$ can be lynchpins of these three weak connections, the total number of vertices in the two branch sets is at most three.  This implies that either $G(f,c)$ or $G(f,d)$ has size one, contradicting the assumptions for this case.

  We first fill $G(f,b)$ with vertices from $G(f,d)$; if we can guarantee that one such vertex has a neighbour in $G(f,c)$, then $y$ will no longer be essential for $G(f,c)$ or $G(f,d)$, and can be relabeled to $a$.

  \smallskip
  \noindent{\bf Case 2a}: $y$ has a neighbour $w$ in $G(f,d)$ such that $w$ is a cut vertex of $G(f,d)$
  \smallskip

  If a neighbour $w$ of $y$ in $G(f,d)$ is a cut vertex, we observe that each leaf block in $G(f,d)$ has at least one interior vertex that has neighbour in $G(f,c)$, as otherwise the joining vertex of the leaf block, $x$, and $y$ form a cut set of size three separating the interior vertices of the leaf block and vertices in $G(f,c)$.  Consequently, we can find a component $C$ of $G(f,d) \setminus \{w\}$ that contains a neighbour of $x$.  

We can then conclude that $f$ does not hit a leaf-crucial model or a leaf-$b$-crucial model on relabeling any component $C' \neq C$ in $G(f,d) \setminus \{w\}$, and hence by Lemma~\ref{lemma-slurp-component}, we can relabel all vertices in $C'$ to labels other than $d$ such that $w$ has a neighbour labeled $c$. Thus siphoning away all components other than $C$ allows $w$ to be relabeled to $b$.  The presence of $w$ in $G(f,b)$ ensures that $y$ is not a cut vertex, and can now be relabeled to $a$.

  \smallskip
  \noindent{\bf Case 2b}: $y$ has no neighbour in $G(f,d)$ that is a cut vertex of $G(f,d)$
  \smallskip

As in Case 1b, if $y$ has a neighbour $w$ in $G(f,d)$ such that $w$ is a neighbour of a vertex in $G(f,c)$, we can relabel $w$ to $b$ in a single step using Lemma~\ref{lemma-conds}, and in another step we can then relabel $y$ to $a$.
  
Otherwise, we use Lemma~\ref{lemma-conds} for the relabeling of neighbour of $y$ to $b$; we can iteratively relabel a neighbour, a neighbour of a neighbour, and so on until we encounter either a vertex that is a neighbour of a vertex in  $G(f,c)$, so that either a neighbour of a vertex in $G(f,c)$ is labeled by $b$ or
Case 2a applies.

  \smallskip
  \noindent {\bf Case 3}: $|G(f,a)| = 1$,  $|G(f,b)| > 1$, and $|G(f,c)| = 1$
  \smallskip

  We first observe that if $|G(f,d)| = 1$, then by $4$-connectivity, each of $G(f,a)$, $G(f,c)$, and $G(f,d)$ have at least two neighbours in $G(f,b)$.  Thus, $y$ is not essential for any branch set, allowing us to achieve the necessary relabeling using either Lemma~\ref{lemma-conds}, if $y$ is a cut vertex in $G(f,b)$, or Lemma~\ref{lemma-slurp-component} otherwise.

  We now assume that $|G(f,d)| > 1$, and observe that as a consequence, there cannot be a weak connection between $G(f,b)$ and $G(f,d)$: by Lemma~\ref{lemma-all-weak}, there will then be weak connections between $G(f,b)$ and all other branch sets, but it will be possible to designate non-lynchpins in both $G(f,b)$ and $G(f,d)$.

 The cases below cover all possibilities in which Lemma~\ref{lemma-conds} does not apply (as otherwise we could relabel $y$ to $a$ in a single step).

    \smallskip
  \noindent {\bf Case 3a}: $y$ is essential for $c$
  \smallskip

  We show that we can fill $G(f,b)$ with vertices that include a neighbour of a vertex in $c$ so that $y$ is no longer essential for $c$.

  By $4$-connectivity, each of $G(f,a)$, $G(f,b)$, and $G(f,c)$ have at least two neighbours in $G(f,d)$.  Each leaf block in $G(f,d)$ must have an interior vertex that is the neighbour of a vertex in $G(f,b)$ (Lemma~\ref{lemma-two-weak} point one) as well as an interior vertex that is the neighbour of a vertex in $G(f,c)$ (as otherwise the joining vertex of the leaf block, $x$, and $y$ form a cut set of size three separating interior vertices in the leaf block and $G(f,c)$).

  We consider two cases, depending on whether or not a vertex in $G(f,b)$ has a neighbour in $G(f,d)$ that is a cut vertex for $G(f,d)$.

  If a vertex in $G(f,b)$ has a neighbour $w$ in $G(f,d)$ that is a cut vertex, since each leaf block in $G(f,b)$ contains neighbours in $G(f,c)$ and $G(f,d)$, we can choose a component $C \in G(f,d) \setminus \{w\}$ to retain, using Lemma~\ref{lemma-slurp-component} to remove all others.  This ensures that $w$ has neighbours in $G(f,c)$ and $G(f,d)$ and is no longer a cut vertex; it can then be relabeled to $b$, so that $y$ is no longer essential for $c$.
  
  If no vertex in $G(f,b)$ has a neighbour in $G(f,d)$ that is a cut vertex, then as in Cases 1b and 2b, vertices can be relabeled iteratively until either a vertex with a neighbour in $G(f,c)$ is labeled $b$, or it reduces to the case above.

  Now $y$ is no longer essential for $c$, either resulting in Case 3b or permitting $y$ to be labeled $a$.
    
  \smallskip
  \noindent {\bf Case 3b}: $y$ is a cut vertex
  \smallskip

  By Lemma~\ref{lemma-two-weak}, point one, each leaf block in $G(f,b)$ must have an 
interior vertex that is the neighbour of a vertex in $G(f,d)$.  We can then identify a component $C \in G(f,b) \setminus \{y\}$ that contains neighbours in both $G(f,c)$ and $G(f,d)$, allowing us to apply Lemma~\ref{lemma-slurp-component} to every component $C' \ne C$.  Then $y$ is no longer a cut vertex and can be relabeled $a$, as needed.

    \smallskip
  \noindent {\bf Case 4}:  $|G(f,b)| > 1$ and $G(f,b)$ does not contain an essential vertex for either $c$ or $d$

  \smallskip

  In this case, the only condition of Lemma~\ref{lemma-conds} that can be violated for the relabeling of $y$ to $a$ is condition~\ref{cond-notcut}, so we assume that $y$ is a cut vertex.

If there is a component of $G(f,b) \setminus \{y\}$ with connecting edges to $G(f,c)$ and $G(f,d)$, then by Lemma~\ref{lemma-slurp-component}, all the other components can be siphoned away, allowing $y$ to be relabeled to $a$.  Otherwise, each component in $G(f,y) \setminus \{b\}$ has either connecting edges to $G(f,c)$ or connecting edges to $G(f,d)$.  We can use Lemma~\ref{lemma-slurp-component} to siphon away all components one at a time, leaving $y$ as the only vertex in $G(f,b)$.  The case has now been reduced to Case 1 or 2 (if $|G(f,a)| = 1$) or, with the roles of $a$ and $b$ reversed, Case 3, 4, 5, or 6 (if $|G(f,a)| > 1$), where the second use of Case 4 cannot again result in Case 4, as at that point the branch sets for $a$ and $b$ will both be of size one.

  \smallskip
  \noindent {\bf Case 5}: $|G(f,a)| = 1$, $|G(f,b)| > 1$, $|G(f,c)| > 1$, $|G(f,d)| > 1$, and $G(f,b)$ has weak connections to $G(f,c)$ and $G(f,d)$
  \smallskip

By applying Lemma~\ref{lemma-all-weak} to $G(f,b)$, we observe that since $x$ can be designated the non-lynchpin of the weak connection between $G(f,b)$ and $G(f,a)$, each vertex in $G(f,b)$ must be essential for $c$, $d$, or both.  Thus, $G(f,b)$ contains at most two vertices; since by assumption $|G(f,b)| > 1$,
  $|G(f,b)| = 2$.

  Without loss of generality, $y$ is essential for $d$ and that $y' \in G(f,b) \setminus \{y\}$ is essential for $c$.  Since we have assumed that Lemma~\ref{lemma-minislurp} does not apply, $G(f,b) \setminus \{y\}$ does not contain the endpoint of any essential edge, and hence $y'$ has at least two neighbours in $G(f,c)$.

 We first show that it is not possible for there to be a weak connection between $G(f,c)$ and $G(f,d)$.  Suppose instead that there were such a weak connection, and hence weak connections between $G(f,a)$ and $G(f,b)$, $G(f,b)$ and $G(f,c)$, $G(f,c)$ and $G(f,d)$, and $G(f,d)$ and $G(f,a)$.  Because $G(f,c)$ contains at least two vertices, only one of which can be designated a lynchpin for the weak connection with $G(f,d)$, $G(f,c)$ contains a non-lynchpin.  But by designating $x$ as the lynchpin for the weak connection between $G(f,a)$ and $G(f,b)$, $y$ is a non-lynchpin in $G(f,b)$.  Hence, by Lemma~\ref{lemma-four-weak}, there cannot be a weak connection between $G(f,c)$ and $G(f,d)$.

  We now assume that there is no weak connection between $G(f,c)$ and $G(f,d)$, and show that we can fill $G(f,b)$ with vertices from $G(f,c)$ with neighbours in $G(f,d)$, allowing us to relabel $y$ by $a$.  The following two cases cover all possibilities:

    \smallskip
  \noindent{\bf Case 5a}: $y'$ has a neighbour $w$ in $G(f,c)$ such that $w$ is a cut vertex of $G(f,c)$
  \smallskip

  By $4$-connectivity, each leaf block in $G(f,c)$ must contain an interior vertex with a neighbour in $G(f,d)$, as otherwise there is a cut set of size three (the joining vertex of the leaf block, $x$, and $y$) separating the interior vertices of the leaf block from the vertices in $G(f,d)$.

  We can then choose a component $C$ of $G(f,c) \setminus \{w\}$ that contains a neighbour of $x$ and a neighbour in $G(f,d)$, and use Lemma~\ref{lemma-slurp-component} to siphon away each other component $C' \ne C$.  After this process, $w$ will have a neighbour in $G(f,d)$ and no longer be a cut vertex.  It can then be relabeled $b$, as needed.

    \smallskip
  \noindent{\bf Case 5b}: $y'$ has no neighbour in $G(f,c)$  that is a cut vertex of $G(f,c)$
  \smallskip

  As in Cases 1b and 2b, if $y'$ has a neighbour $w$ in $G(f,c)$ such that $w$ is a neighbour of a vertex in $G(f,d)$, we can relabel $w$ to $b$ in a single step using Lemma~\ref{lemma-conds}.  Otherwise, we use Lemma~\ref{lemma-conds} repeatedly until we either find such a vertex $w$ as a neighbour of a vertex in $G(f,b)$ or we encounter a cut vertex, to which Case 5a applies.

    \smallskip
    \noindent {\bf Case 6}: $|G(f,b)| > 1$, $|G(f,c)| > 1$, $|G(f,d)| > 1$, and
    there are weak connections between $G(f,a)$ and $G(f,c)$ and between $G(f,b)$ and $G(f,c)$
    \smallskip

    Since Lemma~\ref{lemma-conds} does not hold for the relabeling of $y$ by $a$, $y$ is an essential vertex for $c$, a cut vertex in $G(f,b)$, or both.

    If $y$ is an essential vertex for $c$, we fill $G(f,b)$ with vertices from $G(f,d)$ in order that $y$ is no longer essential for $c$.  By Lemma~\ref{lemma-two-weak}, point two, each leaf block in $G(f,d)$ has interior vertices with neighbours in all three other branch sets.  If a vertex in $G(f,b)$ has a neighbour $w$ in $G(f,d)$ that is a cut vertex, we simply choose a component $C$ of $G(f,d) \setminus \{w\}$ to retain and then siphon away all other components into $G(f,b)$, allowing us to relabel $w$ (which now has neighbours in all other branch sets) to $b$.  Otherwise, we use the process from Cases 1b, 2b, and 5b to relabel vertices in $G(f,d)$ by $b$ until we either have labeled vertices with the required neighbours or we find a cut vertex $w$. We now have a $K_4$-model $g$ in which $y$ is not an essential vertex for $c$.

      Now $y$ is no longer essential for $c$, but it may still be a cut vertex in $G(g,b)$.  If there is a component $C$ of $G(g,b) \setminus \{y\}$ with neighbours in both $G(g,c)$ and $G(g,d)$, we can siphon away all the components $C' \ne C$, and then relabel $y$ with $a$, as needed.  Otherwise, we can siphon away each component in turn, resulting in $y$ being the only vertex with label $b$, and thus reducing to Case 1 or 2 (if $|G(g,a)| = 1$) or Case 3, 4, 5, or 6 (if $|G(g,a)| > 1$), where the second use of Case 6 cannot again result in Case 6, as at that point the branch sets for $a$ and $b$ will both be of size one.   This completes Case 6. 

      In the remainder of the proof, we assume that $|G(f,a)| > 1$ and $|G(f,b)| >1$, and ensure that we consider all possible cases.  We observe that it is not possible to have $|G(f,a)| > 1$, $|G(f,b)| > 1$, and $|G(f,c)| = |G(f,d)| = 1$, since by Lemma~\ref{lemma-four-weak} as $G(f,a)$ and $G(f,b)$ can each contain a non-lynchpin in the weak connections between branch sets with labels $a$ and $b$, $b$ and $c$, $c$ and $d$, and $d$ and $a$.

\smallskip
  \noindent {\bf Case 7}:  $|G(f,a)| > 1$, $|G(f,b)| > 1$, $|G(f,c)| > 1$, and $|G(f,d)| = 1$.

  \smallskip

  In this case, by Lemma~\ref{lemma-all-weak} there cannot be a weak connection between $G(f,a)$ and $G(f,c)$, as we can designate a non-lynchpin in $G(f,a)$ by selecting $y$ and the vertex in $G(f,d)$ as lynchpins of the weak connection between $G(f,a)$ and $G(f,b)$ and the weak connection between $G(f,a)$ and $G(f,d)$, respectively, ensuring that at least one vertex in $G(f,a)$ and at least one vertex in $G(f,c)$ is left as a non-lynchpin.  The same argument can be used to show that there no weak connection between $G(f,b)$ and $G(f,c)$.

  Since $x$ cannot be labeled $b$ in one step, it must be essential for $G(f,d)$, a cut vertex in $G(f,a)$ or both.  Similarly, since $y$ cannot be labeled $a$ in one step, it must be essential for $G(f,d)$, a cut vertex in $G(f,b)$, or both.  We consider two cases, depending on whether or not both $x$ and $y$ are essential for $G(f,d)$.

\smallskip
  \noindent {\bf Case 7a}: $x$ and $y$ are both essential for $G(f,d)$
  \smallskip

By Lemma~\ref{lemma-two-weak}, point two, each leaf block in $G(f,c)$ has interior vertices that are endpoints of connecting edges to both $G(f,a)$ and $G(f,b)$, and by $4$-connectivity, an interior vertex that is an endpoint of a connecting edge to $G(f,d)$ (as otherwise the joining vertex of the leaf block, $x$, and $y$ form a cutset of size three separating interior vertices of the leaf block and $G(f,d)$).  

Our goal is to fill $G(f,b)$ with vertices from $G(f,c)$ so that $y$ is no longer essential for $d$.  If there is a vertex $w \in G(f,b)$ that has a neighbour $z$ which is a cut vertex of $G(f,c)$, we can use Lemma~\ref{lemma-slurp-component} to siphon away all but one component of $G(f,c) \setminus \{z\}$, and then relabel $z$ to $b$.  The remaining component will have all the required connecting edges, and $z$ will have neighbours in all branch sets.  Now either $y$ can be relabeled to $a$ or, if $y$ is still a cut vertex of the branch set with label $b$, we have reduced this case to Case 7b.

If instead no vertex in $G(f,b)$ has a neighbour which is a cut vertex of $G(f,c)$, we can use the technique of Cases 1b, 2b, and 5b.

\smallskip
  \noindent {\bf Case 7b}: $y$ is not essential for $G(f,d)$
  \smallskip

If there exists a component of $G(f,b) \setminus \{y\}$ with endpoints of connecting edges to both $G(f,c)$ and $G(f,d)$, we use Lemma~\ref{lemma-slurp-component} to siphon away all other components, after which $y$ can be relabeled $a$. Otherwise, each component has connecting edges to either $G(f,c)$ or $G(f,d)$.  We can then use the same lemma to siphon away all vertices except $y$, reducing the case to Case 3, 4, 5, or 6.  This completes Case 7.

  We now assume that all branch sets are of size greater than one and consider possibilities for weak connections.  Since we can apply Case 4 with either $a$ or $b$ playing the role of $b$, we can assume that each of $G(f,a)$ and $G(f,b)$ contain an essential vertex for either $c$ or $d$, or, stated in a weaker form, that each of $G(f,a)$ and $G(f,b)$ have a weak connection to either $G(f,c)$ or $G(f,d)$.  Since Case 6 handles the case in which $G(f,a)$ and $G(f,b)$ both have weak connections to the same branch set, Case 8 suffices to complete the proof.

      \smallskip
  \noindent {\bf Case 8}: Each branch set has size greater than one, and there are weak connections between $G(f,b)$ and $G(f,c)$ and between $G(f,a)$ and $G(f,d)$
  \smallskip

We first show that there cannot be a weak connection between $G(f,c)$ and $G(f,d)$.  Since by assumption $G(f,b) \setminus \{y\}$ does not contain the endpoint of an essential edge, either $y$ is the endpoint of an essential edge $\mbox{ess}(b,c)$ or $G(f,c)$ has at least two vertices that are neighbours of vertices in $G(f,b)$.  We show that neither case can hold.

  If $y$ is the endpoint of an essential edge $\mbox{ess}(b,c)$, we can designate the endpoint in $G(f,c)$ as the lynchpin of the weak connection and $y$ as the lynchpin of $xy$.  Since $G(f,a)$ contains at least two vertices, only one of which can be the lynchpin for the weak connection between $G(f,a)$ and $G(f,d)$, $G(f,a)$ must contain a non-lynchpin.  Because $G(f,b)$ also contains a non-lynchpin (any vertex other than $y$), by Lemma~\ref{lemma-four-weak}, this case cannot hold.  

  If instead $G(f,c)$ has at least two vertices that are neighbours of vertices in $G(f,b)$, then $G(f,c)$ contains a non-lynchpin, as it can contain only one vertex that is a lynchpin for the weak connection between $G(f,c)$ and $G(f,d)$.  By making $x$ the lynchpin of $xy$, $G(f,b)$ must also contain a non-lynchpin (any vertex that is not the essential vertex for $c$).  By Lemma~\ref{lemma-four-weak}, this case cannot hold.

We can now assume that there is no weak connection between $G(f,b)$ and $G(f,d)$. Moreover, since Lemma~\ref{lemma-conds} does not hold for the relabeling of $y$ by $a$, $y$ is an essential vertex for $c$, a cut vertex in $G(f,b)$, or both. We show how to handle each possible situation.

  If $y$ is an essential vertex for $c$, we first fill $G(f,b)$ with vertices from $G(f,c)$ such that is $y$ is no longer essential for $c$.  If $y$ has a neighbour that is not a cut vertex, then we can use Lemma~\ref{lemma-conds} to label the neighbour $b$.  If instead every neighbour of $y$ is a cut vertex in $G(f,c)$, we consider one such neighbour $w$ and note that by Lemma~\ref{lemma-two-weak} point two, each leaf block in $G(f,c)$ has an interior vertex that has a neighbour in $G(f,a)$.  We can then choose a component $C$ of $G(f,c) \setminus \{w\}$ that has neighbours in both $G(f,a)$ and $G(f,d)$, and use Lemma~\ref{lemma-slurp-component} to siphon away each other component $C' \ne C$.  After this process, $w$ will no longer be a cut vertex, so it can then be relabeled $b$, forming the $K_4$-model $g$.

  If $y$ is a cut vertex in $G(g,b)$, we use the same technique as in the proof of Cases 4, 6, and 7b to reduce the case to Case 3, 4, 5, or 6.
  
\end{proof}

Lemma~\ref{lemma-split-4} follows the structure of the proof of Lemma~\ref{lemma-split}, relying on Lemma~\ref{lemma-split-samelabel} for the reconfiguring between $K_4$-models in which $x$ and $y$ have the same label and on Lemma~\ref{lemma-xyess} for the case in which $xy$ is an essential edge.  The two possible cases for Lemma~\ref{lemma-addedge-4} are $xy$ being an essential edge (handled by Lemma~\ref{lemma-xyess}) and $xy$ being a bridge in a branch set.

\begin{lemma}\label{lemma-split-4}
  Suppose $G$ is a $4$-connected graph such that $G \in \mbox{host}(K_{4})$ and $G'$ is a $4$-connected graph formed from $G$ by splitting
a vertex $v$ into vertices $x$ and $y$. Then $G'$ is in $\mbox{host}(K_{4})$.
\end{lemma}

\begin{proof}

  Our proof follows the structure of the proof of Lemma~\ref{lemma-split}; we know from Lemma~\ref{lemma-split-samelabel} that we can reconfigure between any two $K_4$-models in which $x$ and $y$ have the same label. Here, we show that we can reconfigure from any $K_4$ model of $G'$ to a $K_4$-model in which $x$ and $y$ have the same label.

We consider a $K_4$-model $f'$ of $G'$ such that $f'(x) \ne f'(y)$; without loss of generality, we assume that $f'(x) = a$, $f'(y) = b$, and that the two remaining labels are $c$ and $d$.  If $xy$ is an essential edge, then the result follows from Lemma~\ref{lemma-xyess}.

Suppose instead that $xy$ is not an essential edge. If Lemma~\ref{lemma-conds} holds for either relabeling $x$ to $b$ or relabeling $y$ to $a$, then we can achieve the reconfiguration in a single step. Hence, at least one condition of Lemma~\ref{lemma-conds} is violated for relabeling $x$ to $b$ and $y$ to $a$. If condition 1 is violated for both $x$ and $y$, then $xy$ is an essential edge, which contradicts our assumption. So without loss of generality, we let $|G'(f',a)| > 1$ and attempt to relabel $x$ to $b$, or if that is not possible, $y$ to $a$. By Lemma~\ref{lemma-conds}, we can relabel neither $x$ nor $y$ in a single step if and only if $x$ (respectively, $y$) is a cut vertex or is $b$-crucial ($a$-crucial) or both. We show that we can relabel other vertices so that either one of $x$ and $y$ can be relabeled or $xy$ is an essential edge, which is handled by Lemma~\ref{lemma-xyess}.

  \smallskip
  \noindent{\bf Case 1:} $|G'(f',a)| = 1$
  \smallskip
 
  Since $xy$ is not an essential edge, $|G'(f',b)| > 1$.
 
  \smallskip
  \noindent{\bf Case 1a:} $y$ is a cut vertex of $G'(f',b)$
  \smallskip

  Suppose $y$ is essential for $c$ or $d$. Without loss of generality, let $y$ be essential for $c$. We show that all other vertices labeled $b$ can be relabeled to $d$. Note that since there are necessary connecting edges from $y$, no other vertex labeled $b$ is $d$-crucial. Furthermore, by $4$-connectivity, each leaf block of $G'(f',b)$ contains at least two interior vertices with neighbours labeled $d$ since otherwise $\{x,y\}$ is a cut set of size two. Hence, we can relabel a vertex in a leaf block to $d$ and repeat the process until $y$ is the only vertex in the branch set for $b$, at which point $xy$ is an essential edge.

  Hence, suppose $y$ is not an essential vertex. If a component $C$ of $G'(f',b) \setminus \{y\}$ has edges to both $G'(f',c)$ and $G'(f',d)$, then we use Lemma~\ref{lemma-slurp-component} to siphon away components other than $C$ and then relabel $y$ to $a$. Otherwise, for each component $C$, each leaf block of $C$ must contain at least two interior vertices with neighbours in one of $G'(f',c)$ and $G'(f',d)$ since otherwise the joining vertex and $x$ form a cut set of size two. Therefore, either of those interior vertices can be relabeled and repeating the process, all components of $G'(f',b) \setminus \{y\}$ can be siphoned away to branch sets with labels $c$ and $d$, which results in $xy$ being an essential edge.

  \smallskip
  \noindent{\bf Case 1b:} $y$ is not a cut vertex of $G'(f',b)$
  \smallskip
  
  Then $y$ must be essential for $c$ or $d$ because otherwise we can relabel $y$ to $a$ in a single step. Now, we use the same argument as in Case 1a to relabel interior vertices of the leaf blocks of the branch set for $b$ until either $y$ is the only vertex in the branch set for $b$, in which case $xy$ is an essential edge, or $y$ is a cut vertex of its branch set, which is Case 1a.

   \smallskip
   \noindent{\bf Case 2:} $|G'(f',a)| > 1$
   \smallskip
    
   We reconfigure to a model where $x$ can be relabeled to $b$ or to a model where $x$ is the only vertex labeled $a$ and thus reduce it to Case 1.
    
   \smallskip
   \noindent{\bf Case 2a:} $x$ is a cut vertex of $G'(f',a)$
   \smallskip
    
   Suppose $x$ is essential for $c$ or $d$. Without loss of generality, let $x$ be essential for $c$. Then no other vertex labeled $a$ is $d$-crucial and furthermore each leaf block of each component of $G'(f',a) \setminus \{x\}$ has at least two interior vertices with neighbours in $G'(f',d)$ by $4$-connectivity. Therefore, we can siphon away each component to obtain a model where $x$ is the only vertex labeled $a$.
    
   Suppose $x$ is essential for neither $c$ nor $d$. If there exists a component of $G'(f',a) \setminus \{x\}$ that has edges to both $G'(f',c)$ and $G'(f',d)$, then we use Lemma~\ref{lemma-slurp-component} to siphon away all other components and then relabel $x$ to $a$. Otherwise, we use Lemma~\ref{lemma-slurp-component} to siphon away every component so that $x$ is the only vertex labeled $a$.
    
   \smallskip
   \noindent{\bf Case 2b:} $x$ is not a cut vertex of $G'(f',a)$
   \smallskip
    
   Then $x$ must be essential for $c$ or $d$ because otherwise we can relabel $x$ to $b$ in a single step. Now, we use the same argument as in Case 2a to relabel interior vertices of the leaf blocks of the branch set for $a$ until either $x$ is the only vertex in the branch set for $a$ or $x$ is a cut vertex of its branch set, which is Case 2a.   

\end{proof}

\begin{lemma}\label{lemma-addedge-4}
Suppose $G$ is a $4$-connected graph such that $G \in \mbox{host}(K_{4})$ and $G'$ is formed from $G$ adding an edge $xy$. Then $G' \in \mbox{host}(K_{4})$.
\end{lemma}

\begin{proof}
   As in the proof of Lemma~\ref{lemma-addedge}, it suffices to show that we can reconfigure between any $K_4$-model of $G'$ and a $K_4$-model of $G$, which we handle in two cases, depending on whether $xy$ is an essential edge or $xy$ is a bridge in a branch set.  As the former is covered by Lemma~\ref{lemma-xyess}, only the latter remains.

We now show that if $f(x) = f(y)$ and $xy$ is a bridge in $G'(f,f(x))$, we can reconfigure to a model in which $x$ and $y$ have different labels.  Depending on whether $xy$ is then an essential edge, we have either completed the reconfiguration or we have reduced the situation to the essential edge case.

We assume without loss of generality that $f(x) = f(y) = a$ and observe that the removal of $xy$ separates $G'(f,a)$ into two components $C_1$ (containing $x$) and $C_2$ (containing $y$), each of which contains at least one leaf block. By Lemma~\ref{lemma-leafblock}, each leaf block has at least three interior vertices with neighbours in other branch sets. 
We will show that we can reconfigure to a $K_4$-model in which either $C_1$ or $C_2$ has no vertex with label $a$, so that $xy$ is no longer a bridge. We consider three cases based on the number of different branch sets to which $C_1$ has a connecting edge.

 \smallskip
  \noindent{\bf Case 1:} $C_1$ has a connecting edge to each of $G'(f,b)$, $G'(f,c)$, and $G'(f,d)$.
  \smallskip
  
  As in Case 2a of the proof of Lemma~\ref{lemma-addedge}, since $x$ is a cut vertex, we use Lemma~\ref{lemma-slurp-component} to relabel the vertices of $C_2$, which ensures $xy$ is a connecting edge.
  
  \smallskip
  \noindent{\bf Case 2:} $C_1$ has connecting edges to two of $G'(f,b)$, $G'(f,c)$, and $G'(f,d)$.
  \smallskip
  
  Suppose without loss of generality that $C_1$ has edges to $G'(f,b)$, $G'(f,c)$. Then $C_2$ has an edge to $G'(f,d)$. Now $f$ does not hit a leaf-$d$-crucial or leaf-crucial model on relabeling $C_2$ because $C_1$ has the necessary connecting edges. Hence, by Lemma~\ref{lemma-slurp-component}, we can relabel the vertices of $C_2$, which ensures $xy$ is a connecting edge.
  
  \smallskip
  \noindent{\bf Case 3:} $C_1$ has a connecting edge to exactly one of $G'(f,b)$, $G'(f,c)$, and $G'(f,d)$.
  \smallskip
  
  Suppose without loss of generality that $C_1$ has an edge to $G'(f,b)$. Then $C_2$ has connecting edges to $G'(f,c)$ and $G'(f,d)$, and so the case reduces to Case 2 with the roles of $C_1$ and $C_2$ swapped.
\end{proof}

\section{Conclusions and open questions}\label{sec-conclusions}

We have developed a toolkit for the reconfiguration of minors, and specific results for $H$-models of small cliques $H$.  Our results imply an alternate definition of $2$-connectivity, whereby a graph is $2$-connected if and only if it is in $\mbox{host}(K_2)$. Furthermore, we have shown that every $3$-connected graph is in $\mbox{host}(K_3)$ and that every $4$-connected graph is in $\mbox{host}(K_4)$, provided that it is not in $\mathcal{L}$, where $\mathcal{L} = \{H : H$ is the line graph of an internally $4$-connected cubic graph$\}$.

It remains to be shown whether similar results can be obtained for larger cliques, or for other graphs $H$.  As our results rely on characterizations of $k$-connected graphs, further work is likely to depend on further progress on such results.

As there are alternate ways of defining adjacency relations, further work is needed to determine which definitions are equivalent and for those that are not, what results can be obtained.  In our work, we can view each label as a token; based on this viewpoint, the adjacency relation we have considered can be viewed as {\em Token relabeling (TR)}, changing the label of one vertex in $G$.  Two other possibilities worthy of consideration are {\em Token sliding (TS)}, swapping the labels of two adjacent vertices in $G$, and {\em Token jumping (TJ)}, swapping the labels of any two vertices in $G$.  Both TS~\cite{IDHPSUU11} and TJ~\cite{KMM12} are well-studied for other types of reconfiguration problems, many of which have unlabeled or distinctly labeled tokens.  The use of TS instead of TR is instrumental in handling degree-one vertices in $G$, which otherwise can rarely be relabeled.

Moreover, it is worth considering an alternate formulation in which solutions are considered to be adjacent if one can be formed from another by reassigning labels to vertices according to some permutation on the labels.

Future directions for research include considering other ways of assessing the reconfiguration graph, such as determining its diameter or, in cases in which the reconfiguration graph is connected, to form algorithms that determine whether there is a path between an input pair of solutions.  It remains open how to characterize isolated vertices in the reconfiguration graph, known as
{\em frozen configurations}~\cite{frozen}.

Throughout the paper, we required every vertex of $G$ to be a member of a branch set in an $H$-model.  If instead we considered a subgraph of $G$, a solution might entail the labeling of a subset of the vertices of $G$. We observe that when the number of labels is equal to the number of vertices in $H$, the problem is reduced subgraph isomorphism~\cite{subgraph}.  Alternative mappings can be considered as well, such as topological embedding of one graph in another.

\bibliography{minor.bib}

\end{document}